\documentclass[onecolumn]{IEEEtran}  

\usepackage{subfigure}
\usepackage{amssymb, amsmath, amsfonts}
\usepackage{amsthm}
\usepackage{tikz}
\usepackage{pgfplots}
\pgfplotsset{compat=newest}
\usepackage{mathrsfs}
\usepackage{graphicx}
\usepackage{enumerate}                                          
\usepackage{epstopdf}

\usepackage{algorithm}
\usepackage{algpseudocode}
\graphicspath{ {./} }
\usetikzlibrary{matrix,arrows,fit,backgrounds,mindmap,plotmarks,decorations.pathreplacing}

\newtheorem{theorem}{Theorem}

\newtheorem{lemma}{Lemma}

\newtheorem{corollary}{Corollary}
\newtheorem{remark}{Remark}
\newtheorem{assumption}{Assumption}

\newlength\figureheight
\newlength\figurewidth
\newlength\fwidth

\makeatletter

\newcommand{\Rmnum}[1]{\expandafter\@slowromancap\romannumeral #1@}
\makeatother

\newcommand{\tikzdir}[1]{#1.tikz}
\newcommand{\revised}[1]{\textcolor{black}{#1}}
\newcommand{\todo}[1]{}

\newcommand{\executeiffilenewer}[3]{%
	\ifnum\pdfstrcmp{\pdffilemoddate{#1}}%
	{\pdffilemoddate{#2}}>0%
	{\immediate\write18{#3}}\fi%
}

\DeclareMathOperator*{\argmax}{arg\; max}     
\DeclareMathOperator*{\tr}{tr}     
\DeclareMathOperator{\Cov}{Cov}
\DeclareMathOperator{\rank}{rank}
\DeclareMathOperator{\diag}{diag}
\DeclareMathOperator{\logdet}{log\;det}
\newcommand{\asrightarrow}{\overset{\tiny a.s.}{\rightarrow}}
\newcommand{\asequal}{\overset{\tiny a.s.}{=}}

\title{An Online Approach to Physical Watermark Design}

\author{Hanxiao Liu, Yilin Mo$^\dag$, Jiaqi Yan, Lihua Xie, and Karl H. Johansson
	\thanks{$\dag$: Corresponding Author.}
	\thanks{H. Liu is with the School of Electrical and Electronic Engineering, Nanyang Technological University, Singapore, and the Division of Decision and Control Systems, the School of Electrical Engineering and Computer Science, KTH Royal Institute of Technology, Sweden. Email: {hanxiao001@ntu.edu.sg}.
	}
	\thanks{
		Y. Mo is with the Department of Automation and BNRist, Tsinghua University, China. Email: {ylmo@tsinghua.edu.cn}.
	}
	\thanks{
		J. Yan and L. Xie are with the School of Electrical and Electronic Engineering, Nanyang Technological University, Singapore. Email: {\{jyan004, elhxie\}@ntu.edu.sg}.
	}

	\thanks{
	 K.H. Johansson is with the Division of Decision and Control Systems, the School of Electrical Engineering and Computer Science, KTH Royal Institute of Technology, Sweden. Email: {kallej@kth.se}.
	}
	\thanks{This work is supported by the A*STAR Industrial Internet of Things Research Program, under the RIE2020 IAF-PP Grant A1788a0023, the Knut and Alice Wallenberg Foundation, the Swedish Foundation for Strategic Research, and the Swedish Research Council.}
}

\begin{document} \maketitle
\begin{abstract}
This paper considers the problem of designing physical watermark signals in order to optimally detect possible replay attack in a linear time-invariant system, under the assumption that the system parameters are unknown and need to be identified online. We first provide a replay attack model, where an adversary replays the previous sensor data in order to fool the system. A physical watermarking scheme, which leverages a random control input as a watermark to detect the replay attack, is then introduced. The optimal watermark signal design problem is cast as an optimization problem, which aims to achieve the optimal trade-off between control performance and intrusion detection. An online watermarking design and system identification algorithm is provided to deal with systems with unknown parameters. We prove that the proposed algorithm converges to the optimal one and characterize the almost sure convergence rate. A numerical example and an industrial process example are provided to illustrate the effectiveness of the proposed strategy.
\end{abstract}


\section{Introduction}
\IEEEPARstart{C}{yber}-Physical Systems (CPSs) offer close integration of computational elements and physical processes~\cite{lee2016introduction}. They are defined as systems where \textquotedblleft \textit{physical and software components are deeply intertwined, each operating on different spatial and temporal scales, exhibiting multiple and distinct behavioral modalities, and interacting with each other in a myriad of ways that change with context}\textquotedblright~\cite{bworld}. Such systems play a critical role in large varieties of fields, such as manufacturing, health care, environment control, transportation, etc. Due to their wide applications and critical functions, it is of paramount importance to ensure the secure operation of CPS~\cite{humayed2017cyber,sandberg2015cyberphysical}. Any successful attack on CPS may jeopardize critical infrastructure and people's lives and properties, even threaten national security. In 2010, Stuxnet malware launched a devastating attack on Iranian uranium enrichment facilities~\cite{langner2011stuxnet, falliere2011w32}. This incident raised a great deal of attention to CPS security in recent years~\cite{ani2017review}. 

However, CPS security faces a wide variety of challenges. Cardenas~\textit{et al.}~\cite{cardenas2009challenges} discussed three main challenges and identified unique properties of CPS security compared to traditional IT security. Besides, the physical part of CPS poses new security challenges. Similar discussion can be found in~\cite{neuman2009challenges}. Gollmann and Krotofil~\cite{gollmann2016cyber} pointed out that also people performing security analysis of CPS is a key challenge. The authors argued that it is difficult for people to expertise in both cyber and physical safety and able to appreciate limitations in their own domain. 

\subsection*{Literature Review}

A significant amount of research effort has been devoted to intrusion and anomaly detection algorithms to enhance CPS security. Zimmer~\textit{et al.} ~\cite{zimmer2010time} presented three mechanisms for time-based intrusion detection. The techniques, through bounds checking, were developed in a self-checking manner by the application and through the operating system scheduler. Mitchell and Chen~\cite{mitchell2011hierarchical} proposed a hierarchical performance model and techniques for intrusion detection in CPS. They classified the modern CPS intrusion detection system techniques into two classes: detection technique and audit material. They summarized advantages and disadvantages in~\cite{mitchell2014survey}. Kwon~\textit{et al.}~\cite{kwon2013security} discussed necessary and sufficient conditions {for when} the attacker could be successful without being detected. Their method can be employed to evaluate vulnerability degree of certain CPSs. Corresponding detection and defense methodologies against stealthy deception attacks can be developed. In~\cite{pasqualetti2013attack}, the authors proposed a mathematical framework for CPS and investigated limitations of the monitoring system. Centralized and distributed attack detection and identification monitors were also discussed.

In this paper, we consider the detection problem of replay attacks. In~\cite{mo2009secure,mo2015physical, mo2014detecting}, a replay attack model is defined and its effect on a steady-state control system is analyzed. An algebraic condition is provided on the detectability of the replay attack. For those systems that cannot detect replay attack efficiently, a physical watermarking scheme is proposed to enable the detection of a replay attack. In particular, by injecting a random control signal, the watermark signal, into the control system, it is possible to secure the system. However, the watermark signal may deteriorate the control performance, and therefore it is important to find the optimal trade-off between the control performance and the detection efficiency, which can be cast as an optimization problem. Similar watermarking schemes are also proposed in the literature~\cite{khazraei2017new,Satchidanandan2017, khazraei2017replay}. 

Different from the previous additive watermarking schemes, a multiplicative sensor watermarking scheme is proposed in~\cite{Ferrari2017}. In this scheme, each output is respectively fed to a SISO watermark generator and due to the inclusion of a watermark removing functionality, the control performance will not be sacrificed. Applying some techniques of non-cooperative stochastic games, Miao~\textit{et al.}~\cite{FeiMiao2013} designed a suboptimal switching control policy that balances control performance and the intrusion detection rate for replay attacks. Hoehn and Zhang~\cite{hoehn2016detection} provided a novel technique via exciting the system in non-regular time intervals and signal processing to detect the replay attack. Other replay attack detection mechanisms have also been proposed in the literature~\cite{Shoukry2015}.

It is worth noticing that in majority of the aforementioned research, the precise knowledge of the system parameters is required in order to design the watermark signal and  the detector. However, acquiring these parameters may be troublesome and costly. Moreover, for a large system, the system parameters may change during its operation. Hence, it is beneficial for the system to learn the parameters in an online fashion and automatically generate the optimal detector and the watermark signal in real-time. The problem of learning parameters of dynamical systems, system identification, has been studied over the past decades. Most methods, however, require persistent excitation on the input. 

In this paper, due to the nature of the optimal watermark signal, we shall design the input that asymptotically converges to a signal that does not satisfy the persistent excitation condition. However, by controlling the convergence rate, we can still prove that the system parameters converge to true parameters almost surely.

Some preliminaries results regarding online design of physical watermarks are contained in our former work~\cite{8619632}. The main differences between the current version of the paper and ~\cite{8619632} are: 1) we not only prove that we can asymptotically identify the system parameters, but also characterize the rate of the convergence; 2) we provide a procedure to automatically generate the Neyman-Pearson detector; 3) we add the simulation on an industrial process to verify the effectiveness of the proposed approach. 

\subsection*{Contributions}
The goal of this paper is to develop a data-driven approach to design physical watermark signals to protect systems with unknown parameters, against replay attack. The main contributions of this paper are threefold:
\begin{enumerate}
	\item An online ``learning'' algorithm is presented to simultaneously infer the parameters of the system based only on the system input and output data and generate the watermark signal as well as the optimal detector based on the estimated parameters. To the best of our knowledge, it is the first time to study the detection of replay attacks under the scenario with unknown system parameters. 
	\item We prove that the system parameters which are inferred via our proposed online algorithm converge to the true parameters almost surely even if the input signal asymptotically converges to a degenerate signal. 
	\item We also characterize the almost sure convergence rate of the estimated system parameters to the true parameters and provide an upper bound for this rate. 
\end{enumerate}

\subsection*{Outline of the Paper}
The rest of paper is organized as follows. Section~\ref{sec:problem} formulates the problem by introducing the system as well as the attack model. The physical watermarking scheme is introduced in Section~\ref{sec:watermarking}. In Section~\ref{sec:main}, we present an online algorithm to simultaneously infer the parameters of the system and design the watermark signal as well as the detector based on the estimated parameters. We further prove the almost sure convergence of the watermark signal to the optimal one and characterize the convergence rate. In Section~\ref{sec:simulation}, a numerical example and an industrial process example are provided to verify the effectiveness of the proposed technique. Concluding remarks are given in Section~\ref{sec:conclusion}. For the sake of legibility, most of the proofs are included in the appendix.

\subsection*{Notations}
$\|A\|$ of the matrix $A$ is the spectral norm of an $m\times n$ matrix $A$, which is its largest singular value. $A\otimes B$ is the Kronecker product of matrices $A$ and $B$. $A > 0$ ($A\geq 0$) indicates that $A$ is positive definite (positive semidefinite). \revised{$A^+$ denotes the pseudo-inverse of $ A $.} We say that $f(k)\sim O(g(k))$ if there exists an $M > 0$, such that $|f(k)|\leq M\times g(k)$ for all $k\in \mathbb N_0$.

\section{Problem Formulation}
\label{sec:problem}
In this section, we introduce a linear time invariant system model of CPS as well as a replay attack model, which will be employed in the rest of this paper.

We consider a linear time-invariant system described by the following equation:
\begin{align}
  x_{k} &= A x_{k-1} + B\phi_k + w_{k}, \label{eq:systemdynamic}
\end{align}	
where $x_k\in \mathbb R^{n}$ is the state vector at time $k$, and $w_{k}\in \mathbb{R}^{n}$ is a zero mean independently and identically distributed (i.i.d) Gaussian process noise with covariance $Q \geq 0$. $\phi_k\in\mathbb R^p$ is the watermark signal that will be discussed in details in Section~\ref{sec:watermarking}.

A sensor network is monitoring the above system. The observation equation is given by
\begin{align}
y_{k}  &= C x_k + v_k \label{eq:sensor}, 
\end{align}	
where $y_{k}\in \mathbb{R}^{m}$ is a collection of all sensors' measurements at time $k$. $v_{k}\in \mathbb{R}^{m}$ is a zero mean i.i.d. Gaussian measurement noise with covariance $R \geq 0$. 


\begin{remark}
	To simplify notations, in this paper we consider a stable open-loop system. However, our framework can be easily extended to a closed loop system with an unstable plant but a stabilizing controller, which is discussed in Section~\ref{sec:watermarking}.
	
	Notice that the purpose of the watermark signal is intrusion detection instead of stabilization. As a result, we only consider stable systems or systems that have been pre-stabilized by some controller.
\end{remark}

We assume that the process noise $w_0,w_1,\cdots$ and the measurement noise $v_0,v_1,\cdots$ are independent of each other. Furthermore, since CPSs usually operate for an extended period of time, it is assumed that the system is already in the steady state, which means that the initial condition $x_{-1}$ is a zero mean Gaussian random vector independent of the process noise and the measurement noise and with covariance $\Sigma$, where $\Sigma$ satisfies the following Lyapunov equation:
\begin{align}
\Sigma = A\Sigma A^T+Q.
\label{eq:Sigmadef}
\end{align}

We further make the following assumptions regarding the system parameters:
\begin{assumption}
  The system is strictly stable. Furthermore, $(A,C)$ is observable and $(A,B)$ is controllable.
\end{assumption}

\begin{remark}
  The observability and controllability assumption is without loss of generality as we can perform a Kalman decomposition~\cite{chen1998linear} and only work with the observable and controllable subspace.
\end{remark}

Next we introduce a replay attack model. We assume that the adversary has the following capabilities:
\begin{enumerate}	
\item The attacker has access to all the real-time sensory data. In other words, it knows the sensor's measurement $y_0,\cdots, y_k$ at time $k$.
\item The attacker can modify the real sensor signals $y_k$ to arbitrary sensor signals $ y_k' $.
\end{enumerate}

Given these capabilities, the adversary can employ the following replay attack strategy:
\begin{enumerate}
\item 
  The attacker records a sequence of sensor measurements $y_k$s from time $k_1$ to $k_1+T$, where $T$ is large enough to guarantee that the attacker can replay the sequence for an extended period of time during the attack. 
\item The attacker modifies the sensor measurements $y_k$ to the recorded signals from time $k_2$ to $k_2+T$, i.e.,
  \begin{align*}
    y_{k}' = y_{k-\Delta k},\ \forall\; k_2\leq k\leq (k_2+T), 
  \end{align*}
  where $\Delta k = k_2 - k_1$.
\end{enumerate}

Notice that since the system is already in the steady state, both the replayed signal $y_k'$ and the real signal $y_k$ from the sensors will share exactly the same statistics. As a result, replay attack can be stealthy for a large class of linear systems, if no watermark signal is present, i.e. $\phi_k = 0$. For more detailed discussion on the detectability of replay attack, please refer to~\cite{mo2009secure}.

Let us consider the system illustrated in Fig.~\ref{Fig:system}.
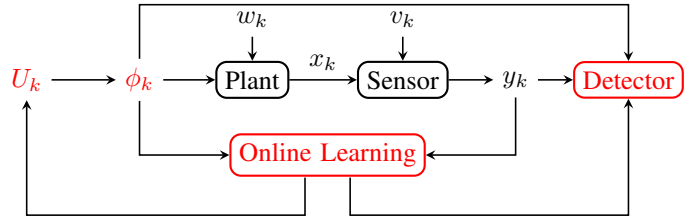
\begin{figure}[htbp]
	\centering
	\begin{tikzpicture}[->,>=stealth',
	box/.style={rectangle, draw=black,rounded corners, semithick},
	point/.style={coordinate},]
	\node [color = red](p11) at (-1.5,0) {$U_k$};
	\node [color = red](p1) at (0,0) {$\phi_k$};
	\node [draw = black, thick,rounded corners] (p2) at (1.5,0) {Plant};
	\node [draw = black, thick,rounded corners] (p3) at (3.5,0) {Sensor};
	\node (p4) at (5,0){$y_k$};
	\node (p5) at (1.5,0.8) {$w_k$};
	\node (p6) at (3.5,0.8) {$v_k$};
	\node (p7)[point] at (5,-1){};
	\node (p8)[point] at (-1.5,-1){};
	\node (p9)[point] at (0,1){};
	\node (p15)[point] at (0,-1){};
	\node (p16)[point] at (-1.5,-1.8){};
	\node (p22)[point] at (6.5,-1.8){};
	\node (p18)[point] at (2.2,-1.8){};
		\node (p21)[point] at (2.8,-1.8){};
	\node (p19)[point] at (2.2,-1.3){};
	\node (p20)[point] at (2.8,-1.3){};
	\node (p10)[point] at (-1.5,1){};
	\node (p13)[point] at (6.5,1){};
	\node [draw = red, thick,rounded corners,color = red] (p12) at (6.5,0) {Detector};
	\node [draw = red, thick,rounded corners,color = red] (p14) at (2.5,-1) {Online Learning};
	\draw [semithick,-stealth] (p11)--(p1);
	\draw [semithick,-stealth] (p1)--(p2);
	\draw [semithick,-stealth] (p2)--node[above]{$x_k$}(p3);
	\draw [semithick,-stealth] (p3)--(p4);
	\draw [semithick,-stealth] (p5)--(p2);
	\draw [semithick,-stealth] (p6)--(p3);
	\draw [semithick,-stealth] (p1)--(p15)--(p14);
	\draw [semithick,-stealth] (p4)--(p7)--(p14);
	\draw [semithick,-stealth] (p19)--(p18)--(p16)--(p11);
		\draw [semithick,-stealth] (p20)--(p21)--(p22)--(p12);
	\draw [semithick,-stealth] (p1)--(p9)--(p13)--(p12);
	\draw [semithick,-stealth] (p4)--(p12);
	\end{tikzpicture}
	\caption{The system diagram.}
	\label{Fig:system}
\end{figure}

\revised{The overarching goal of this paper is to design an online learning algorithm for the optimal replay attack detector as well as the optimal parameters $U_k$ of the physical watermark signals, based on the collected input $\phi_k$ and output $y_k$. The physical watermark scheme is introduced in detail in Section~\ref{sec:watermarking}. Based on this scheme, we develop an approach to infer the system parameters based only on the system input data $\phi_{k}$ and output data $y_k$, and design the highlighted parameters in Fig.~\ref{Fig:system}: the covariance $U_k$ of the watermark signal $\phi_k$ and the optimal detector based on the estimated parameters.}
  
\section{Physical Watermark for Systems with Known Parameters}
\label{sec:watermarking}
This section introduces the concept of physical watermark, which enables the detection of replay attack. The optimal physical watermark is derived via solving an optimization problem which aims to achieve the optimal trade-off between control performance and intrusion detection. Then we will present the extension to a closed-loop system.

\subsection{Physical Watermark Scheme}
The main idea of physical watermark is to inject a random noise $\phi_k$, which is called the watermark signal, into the system \eqref{eq:systemdynamic} to excite the system and check whether the system responds to the watermark signal in accordance to the dynamical model of the system. In this section we will restrict the watermark signal $\phi_k$ to be zero mean i.i.d. Gaussian random variables and its covariance is denoted as $U$.

In the absence of the attack, $y_{k}$ can be represented as:
\begin{align} 
  y_{k}&=\sum_{t=0}^{k} CA^{t}B  \phi_{k-t} + \sum_{t=0}^{k}  CA^{t} w_{k-t}+v_{k} + CA^{k+1} x_{-1}.
         \label{eq:yexpansion}
\end{align}
For simplicity, let us define
\begin{align}
  \varphi_{k} &\triangleq \sum_{\tau=0}^{k} H_\tau  \phi_{k-\tau}\label{eq:varphidef},\, \vartheta_k \triangleq \sum_{t=0}^{k}  CA^{t} w_{k-t}+v_{k} + CA^{k+1} x_{-1},
\end{align}
where $H_\tau$ is defined as
\begin{align}
  H_\tau\triangleq CA^\tau B.
  \label{eq:Hdef}
\end{align}
Therefore, $y_{k}$ can be simplified as:
\begin{align} 
  y_{k}=\varphi_{k} + \vartheta_{k}. \label{eq6}
\end{align}
It is easy to show that $\varphi_{k}$ is a zero mean Gaussian whose covariance converges to $\mathcal U$, where
\begin{align}
  \mathcal U \triangleq \sum_{\tau=0}^\infty H_\tau UH_\tau^T.
  \label{eq:Udef}
\end{align}
Similarly, $\vartheta_{k}$ is a zero mean Gaussian noise whose covariance is $\mathcal W = C\Sigma C^T+R$, where $\Sigma$ is defined in \eqref{eq:Sigmadef}.

On the other hand, let us consider the system under the replay attack, where the replayed $y_k'$ can be written as
\begin{align*}
  y_{k}' = y_{k-\Delta k}= \varphi_{k-\Delta k} + \vartheta_{k-\Delta k},
\end{align*}
Now since $\Delta k$ is unknown to the system operator, we shall treat $\varphi_{k-\Delta k}$ as a zero mean Gaussian random variable with covariance $\mathcal U$. As a result, $y_k'$ is a zero mean Gaussian random variable with covariance $\mathcal U+ \mathcal W$. Therefore, to detect replay attack, we need a detector to differentiate the distribution of $y_k$ under the following two hypotheses:
\begin{enumerate}
\item[$\mathcal H_0$:] The sensor measurement $y_{k}$ follows a Gaussian distribution $ \mathcal{N}_0(\varphi_{k}, \mathcal W) $.
\item[$\mathcal H_1$:] The sensor measurement $ y_{k} $ follows a Gaussian distribution $ \mathcal{N}_1(0, \mathcal U + \mathcal W) $.
\end{enumerate}

\begin{remark}
\revised{It is worth noticing that the watermark signal $\phi_0, \cdots, \phi_{k}$ are known to the system operator and detector and the conditional distribution (conditioned on $\{\phi_k\}_k$) of $y_k$ converges to a Gaussian distribution with mean $\varphi_{k}$ and covariance $\mathcal W$. }
\end{remark}
The Neyman-Pearson detector~\cite{scharf1991statistical} for hypothesis $\mathcal H_0$ versus hypothesis $\mathcal H_1$ takes the following form:
\begin{lemma}\label{theorem:npdetector}
  At time $k$, the Neyman-Pearson detector rejects $\mathcal H_0$ in favor of $\mathcal H_1$ if
  \begin{equation}\label{eq11}
      g_k  = \big(y_{k}- \varphi_{k}\big)^{{T}}  \mathcal W^{-1}\big(y_{k}- \varphi_{k}\big)- y_{k}^{{T}}  \left(\mathcal W+\mathcal U\right)^{-1}y_{k}\\
      \geq \eta,
  \end{equation}
  where \revised{$\eta$ is a threshold chosen by the system operator.} Otherwise, hypothesis $ \mathcal H_{0} $ is accepted.
\end{lemma}

\begin{remark}
  For simplicity, we only consider detecting replay attack based on the current measurement $y_k$. In principle, one may take a moving horizon approach to design a detector, by considering joint distribution of $y_k,\,y_{k-1},\cdots, y_{k-\Delta t}$. However, the proposed methodology in this paper can be easily extended to multiple $y_k$s case by stacking the state vector.
\end{remark}

\begin{remark}
  \revised{It is worth noticing that since hypothesis $\mathcal H_0$ is time-varying due to the $\varphi_k$ term, the threshold $\eta$ needs to be time-varying to ensure a constant false alarm rate. If $\eta$ is still chosen as a constant instead, then the system operator could calculate the expected false alarm rate by numerical integration, since $\varphi_k$ is a stationary process.}
\end{remark}

The following theorem quantifies the performance of the detector, in terms of the expected KL-divergence between distribution $\mathcal N_0$ and $\mathcal N_1$:
\begin{theorem}\label{theorem 2_1}
  The expected KL divergence of distribution $ \mathcal{N}_{0} $ and $ \mathcal{N}_{1} $ is 
  \begin{equation}\label{eq28}
    \mathbb{E}\ D_{KL}\left(\mathcal{N}_{1}\|\mathcal{N}_{0} \right)=\tr\left(\mathcal U \mathcal W^{-1}  \right) - \frac{1}{2}\logdet\left( I+\mathcal U \mathcal W^{-1} \right).
  \end{equation}
  Furthermore, the expected KL divergence satisfies the inequality
  \begin{equation}\label{eq29}
    \begin{split}
      \frac{1}{2}\tr\left(\mathcal U \mathcal W^{-1}  \right)&\leq \mathbb{E}\ D_{KL}\left(\mathcal{N}_{1}\|\mathcal{N}_{0} \right)\\
      & \leq \tr\left(\mathcal U \mathcal W^{-1}  \right) - \frac{1}{2}\log\left[ 1+\tr\left(\mathcal U \mathcal W^{-1}\right) \right].
    \end{split}
  \end{equation}
\end{theorem}
\begin{proof}
  The proof is essentially the same as the proof in~\cite{mo2015physical}. 
\end{proof}
\begin{remark}
  It is worth noticing that the expected KL-divergence is a convex function of $\mathcal U$ and hence $U$. However, both the upper and lower bounds of it are increasing functions of $\tr(\mathcal U\mathcal W^{-1})$. Hence, instead of directly maximizing the detection performance, which is computationally difficult, we could maximize $\tr(\mathcal U\mathcal W^{-1})$, which is linear with respect to $U$.
\end{remark}

Note that although the watermark signal can enable the detection of replay attack, it also deteriorates the system control performance. As a result, it is important to design the signal to achieve the optimal trade-off between the control performance loss and the detection performance. In this paper, to quantify the performance loss, we use the following Linear Quadratic Gaussian (LQG) metric:
\begin{equation}\label{eq3}
  J = \lim_{T \to +\infty} \mathbb E \left(\frac{1}{T}\sum_{k=0}^{T-1} \begin{bmatrix}
      y_k\\
      \phi_k
    \end{bmatrix}^TX \begin{bmatrix}
      y_k\\
      \phi_k
    \end{bmatrix} \right),
\end{equation}
where 
\begin{align*}
  X  = \begin{bmatrix}
    X_{yy}&X_{y\phi}\\
    X_{\phi y}&X_{\phi\phi}\\
  \end{bmatrix} > 0
\end{align*}
\revised{is the weight matrix for the LQG control, which is chosen by the system operator.
\begin{remark}
	The LQG cost is a common choice to quantify the performance of a system running in steady state. On the other hand, we do not foresee any fundamental difficulty to incorporate other performance metrics into our framework, as long as they can be computed from the Markov parameters $H_\tau$.
\end{remark}}

Since $y_k$ and $\phi_k$ converge to a stationary process, $J$ can be written in an analytical form as
\begin{align*}
  J = \lim_{k\rightarrow}\tr\left(X\Cov\left(\begin{bmatrix}
        y_k\\
        \phi_k 
      \end{bmatrix}
  \right) \right) =\tr\left(X \begin{bmatrix}
    \mathcal W + \mathcal U & H_0U\\
    UH_0^T & U
  \end{bmatrix} \right).
\end{align*}
Therefore, $J$ is an affine function of $U$, which can be written as
\begin{align*}
  J = J_0 + \Delta J = \tr(X_{yy}\mathcal W) + \tr(XS),
\end{align*}
where $J_0$ is the optimal LQG cost, and $S$ is linear with respect to $U$, being defined as
\begin{align*}
  S \triangleq \begin{bmatrix}
    \mathcal U & H_0U\\
    UH_0^T & U
  \end{bmatrix}.
\end{align*}

Therefore, in order the achieve the optimal trade-off between the control performance and detection performance, we can formulate the following optimization problem: 
\begin{align}
  U_* =  & \argmax_{U\geq 0} & & \tr (\mathcal U\mathcal W^{-1})\nonumber\\
       & \text{subject to} & &  \tr(XS) \leq \delta,\label{eq:opt}
\end{align}
where $\delta$ is a design parameter depending on how much control performance loss is tolerable.

An important property of the optimization problem~\eqref{eq:opt} is that the optimal solution is usually a rank-$1$ matrix, which is formalized by the following theorem:
\begin{theorem}
  \label{theorem:rankone}
  The optimization problem~\eqref{eq:opt} is equivalent to
  \begin{align}
    U_* =  & \argmax_{U\geq 0} & & \tr (U\mathcal P)\nonumber\\
         & \text{subject to} & &  \tr(U\mathcal X) \leq \delta,\label{eq:opt2}
  \end{align}
  where
  \begin{align}
    \mathcal P &\triangleq\sum_{\tau=0}^\infty H_\tau^T\mathcal W^{-1}H_\tau,\\
    \mathcal X &\triangleq\left(\sum_{\tau=0}^\infty H_\tau^T X_{yy}H_\tau \right)+ H_0^TX_{y\phi} + X_{\phi y}H_0+ X_{\phi \phi}.
  \end{align}
  The optimal solution to \eqref{eq:opt2} is 
  \begin{align*}
    U_* = zz^T,
  \end{align*}
  where $z$ is the eigenvector corresponding to the maximum eigenvalue of the matrix $\mathcal X^{-1}\mathcal P$ and $z^T\mathcal Xz = \delta$. Furthermore, the solution is unique if $\mathcal X^{-1}\mathcal P$ has only one maximum eigenvalue.
\end{theorem}
\begin{proof}
  From the definition of $\mathcal U$, we know that
  \begin{align*}
    \tr(\mathcal U\mathcal W^{-1})  &= \sum_{\tau=0}^\infty\tr\left(H_\tau UH_\tau^T\mathcal W^{-1}\right)\\
                                    &= \sum_{\tau=0}^\infty\tr\left(UH_\tau^T\mathcal W^{-1}H_\tau\right)= \tr\left(U\mathcal P\right).
  \end{align*}
  Following similar steps as in the above proof, we have that $\tr(XS) = \tr(U\mathcal X)$. Moreover, since $X>0$, we have that
\begin{align*}
  \mathcal X &\geq H_0^TX_{yy}H_0 + H_0^TX_{y\phi}+X_{\phi y} H_0+ X_{\phi \phi} \\
             & \geq X_{\phi\phi} - X_{\phi y}X_{yy}^{-1}X_{y\phi} > 0.
\end{align*}
  If the optimal $U_*$ has rank greater than $1$, then follow the same line of argument in the proof of Theorem~7 in \cite{mo2014detecting}, $U_*$ can be decomposed as $U = \alpha_1 U_1+\cdots \alpha_l U_l$, where the following holds
  \begin{enumerate}
  \item $\alpha_i > 0$, $\sum_{i=1}^l \alpha_i = 1$.
  \item $U_i \geq 0$ is of rank $1$ and $\tr(U_i\mathcal X) = \delta$ for all $i$.
  \end{enumerate}
  
  Therefore, by the optimality of $U_*$, we can conclude that
  \begin{align*}
   \tr(U\mathcal P)\leq \min_{i=1,\ldots,l}\tr(U_i\mathcal P). 
  \end{align*}
  However, since $U_*$ is a convex combination of $U_1,\ldots,U_l$, we must have
  \begin{align*}
    \tr(U_*\mathcal P) = \tr(U_1\mathcal P) = \cdots = \tr(U_l\mathcal P),
  \end{align*}
  which shows that the rank one matrix $U_i$ is also optimal.
  
  In order to derive the optimal rank one $U_*$, it is clear that $U_*= zz^T$ for some $z \neq 0$. Hence, the optimization problem~\eqref{eq:opt2} is converted to
  \begin{align*}
    z =&  \argmax_{z\neq 0} & & z^T\mathcal P z\\
       &\text{subject to} & &z^T\mathcal Xz \leq \delta.
  \end{align*}
  Using the Lagrangian multipliers, one can prove that $\mathcal Pz = \lambda \mathcal X z$, which shows that $z$ is the eigenvector of $\mathcal X^{-1}\mathcal P$. If we enumerate all eigenvectors of $\mathcal X^{-1}\mathcal P$, it is not difficult to prove that the maximum is achieved when $z$ is the eigenvector corresponds to the largest eigenvalue of $\mathcal X^{-1}\mathcal P$ and $z^T\mathcal X z = \delta$. 
\end{proof}

\subsection{Extension to Closed-loop Systems}
Before continuing on to the next section, we would like to discuss how to generalize the problem formulation for a closed-loop system with a stabilizing controller. Consider the following system discussed in~\cite{mo2009secure}: 
\begin{align*}
  x_{k+1} &= A x_k + B (u_k + \phi_k) + w_k,\; y_k = Cx_k + v_k, 
\end{align*}
with the following estimator and controller:
\begin{align*}
  \hat x_{k+1} &= A\hat x_k + K(y_{k+1}-CA\hat x_k),\; u_k = L \hat x_k,
\end{align*}
and LQG cost as
\begin{align*}
  J = \lim_{T\to\infty} \frac{1}{T}\mathbb E \left [\sum_{k=0}^{T-1} y_k^T X_{yy} y_k + (u_k+\phi_k)^TX_{\phi\phi}(u_k+\phi_k)\right],
\end{align*}
where $u_k$ denotes the optimal LQG control signal.


We can redefine the state $\tilde x_k$ and output $\tilde y_k$ as
\begin{align*}
  \tilde x_k = \begin{bmatrix}
    x_k\\
    \hat x_k
  \end{bmatrix} ,\text{ and } \tilde y_k = \begin{bmatrix}
    y_k\\
    u_k 
  \end{bmatrix},
\end{align*}
and the design of watermark signal in a closed-loop system can be converted to the open-loop formulation.

It is worth noticing that in order to design the detector and the optimal watermark signal, precise knowledge of the system parameters is needed. However, acquiring the parameters may be troublesome and costly. Furthermore, there may be unforeseen changes in the model of the system, such as topological changes in power systems. As a result, the identified system model may change during the system operation. Therefore, it is beneficial for the system to ``learn'' the parameters and design the detector and watermark signal in real-time, which will be our focus in the next section.

\section{Physical Watermark for Systems with Unknown Parameters}
\label{sec:main}
This section is devoted to developing an online ``learning'' procedure to infer the system parameters, based on which, we show how to design watermark signals and the optimal detector and prove that the physical watermark and the detector asymptotically converge to the optimal ones. 

Throughout the section, we make the following assumptions:
\begin{assumption}\label{as2}
  \begin{enumerate}
  \item $A$ is diagonalizable.
  \item The maximum eigenvalue of $\mathcal X^{-1}\mathcal P$ is unique.
  \item The system is not under attack during the learning phase.
  \item The number of distinct eigenvalues of $A$, which is denoted as $\tilde n$, is known.
  \item The LQG weight matrix $X$ and the largest tolerable LQG loss $\delta$ are known.
  \end{enumerate}
\end{assumption}

\begin{remark}
The first and second assumptions are required in order to ensure that the optimal covariance of the watermark signal is a differentiable function of $H_\tau$, \revised{i.e., the problem is not ill-conditioned. The third assumption is necessary since there is no way to do system identification without (real) sensory data and it is also needed to prove the asymptotic convergence of our algorithm to the true optimal solution as this cannot be achieved in finite time due to the inherent process and measurement noise.} Nevertheless, we shall illustrate through simulation, that after a certain period of learning phase, our algorithm can approximate the optimal solution with reasonably well accuracy and the system can detect replay attack. The fourth assumption is also required to prove convergence, although we shall demonstrate in the simulation that we can use a reduced model to approximate the system with good accuracy. The fifth assumption should hold for all practical cases as $X$ and $\delta$ are design parameters chosen by the system operator.
\end{remark}

For the sake of legibility, we shall introduce our algorithm first and present the theorem on the correctness of our approach in the end.

\subsection{An Online Algorithm}
In this subsection, we will present the complete algorithm in a pseudo-code form. After that, the online ``learning'' scheme will be introduced in detail.

Algorithm~\ref{euclid} describes our proposed online watermarking algorithm. The notations are described later in the subsection.

First, we initialize some parameters which will be used later. In each round of the \textbf{while} iteration, the optimal covariance of the watermarking $U_{k,*}$ based on current knowledge is computed firstly. Based on the derived covariance, one can update the covariance $U_{k}$ by combining ``exploration'' and ``exploitation'' term which will be described in detail later. According to the updated covariance, we generate the watermarking signals $\phi_{k}$ and inject them to the plant. Then we collect the sensory data $y_k$ and employ them and watermarking signals to infer necessary system parameters $H_{k,\tau}, \mathcal P_k, \mathcal X_k$. Based on the estimated parameters, one can update the Neyman-Pearson detector $\hat g_k$. Then one can repeat the above process to identify system parameters and design the watermarking signals as well as the optimal detector.

A pseudo-code form for Algorithm~\ref{euclid} is as follows:

\begin{algorithm}
	\caption{Online Watermarking Design}\label{euclid}
	\begin{algorithmic}[1]
		\Require $\mathcal P_{-1}\gets I,\,\mathcal X_{-1}\gets X_{\phi\phi}, k \gets 0$
		\Ensure
		\While{true}
		\State $U_{k,*} \gets \argmax_{U\geq 0,\,\tr(U\mathcal X_{k-1})\leq \delta} \tr(U\mathcal P_{k-1})$
		\State $U_k \gets U_{k,*} + (k+1)^{-\beta}\delta I $
		\State Generate random variable $\zeta_k\sim\mathcal N(0,I)$ 
		\State Apply watermark signal $\phi_k \gets U_k^{1/2}\zeta_k$
		\State Collect sensory data $y_k$
		\State $H_{k,\tau} \gets \frac{1}{k-\tau+1} \sum_{t=\tau}^k y_t \phi_{t-\tau}^T U_{t-\tau}^{-1}$
		\State Compute the coefficient of $p_k(x)$ by solving \eqref{eq:optcharpoly}
		\If{$p_k(x)$ is Schur stable}
		\State Update $\mathcal P_k,\mathcal X_k$ from \eqref{eq:Omegaest}-\eqref{eq:Xest}
		\EndIf 
		
		\State Update $\hat g_k$ from \eqref{eq:gest}
		\State $k \gets k+1$
		\EndWhile\label{euclidendwhile}
	\end{algorithmic}
\end{algorithm}
\begin{remark}
	For Algorithm~\ref{euclid}, $\mathcal{P}_k, \mathcal{X}_k$ are defined in~\eqref{eq:learningPX}, $U$ is the covariance of watermarking signal, and $H_{k,\tau}$ is defined in~\eqref{eq:Htauest}. Step 3 is the update of the covariance of the physical watermark in~\eqref{eq:Ukdef}. All parameters will be illustrated in the following subsections.
\end{remark}

Then we will introduce this algorithm in detail.

\subsection*{Generation of the Watermark Signal $\phi_k$}
Let us design $U_k$, which can be considered as an approximation for the optimal covariance of the watermark signal $U$, as 
\begin{align}
  U_{k} = U_{k,*} + \frac{\delta}{(k+1)^\beta} I,
  \label{eq:Ukdef}
\end{align}
where $0<\beta<1$, $\delta$ is the maximum tolerable LQG loss defined in \eqref{eq:opt}, and $U_{k,*}$ is the solution of the following optimization problem
\begin{align}
  U_{k,*} =  & \argmax_{U\geq 0} & & \tr (U\mathcal P_{k-1}),\nonumber\\
             & \text{subject to} & &  \tr(U\mathcal X_{k-1}) \leq \delta,\label{eq:learningPX}
\end{align}
and $\mathcal P_{k-1}$ and $\mathcal X_{k-1}$ are the estimate of $\mathcal P$ and $\mathcal X$ matrices, respectively, based on $y_0,\ldots, y_{k-1},\phi_0,\ldots, \phi_{k-1}$, both of which are initialized as:
\begin{align*}
\mathcal P_{-1} = I,\,\mathcal X_{-1} = X_{\phi\phi}.
\end{align*}
The inference procedure of $\mathcal P_k$ and $\mathcal X_k$ for $k \geq 0$ will be provided in the further subsections.

\begin{remark}
Notice that the second term $(k+1)^{-\beta} I$ on the RHS of \eqref{eq:Ukdef} is crucial for parameter identification. The reason is that $U_{k,*}$ is in general a rank $1$ matrix (as is proved in Thereom~\ref{theorem:rankone}) and hence it does not provide persistent excitation to the system for us to identify the necessary parameters. Conceptually, the $(k+1)^{-\beta}I$ term can be interpreted as an ``exploration'' term, as it provide necessary excitation to the system in order for us to infer the parameters. The $U_{k,*}$ is the ``exploitation'' term, as it is optimal under our current knowledge of the system parameters.
\end{remark}

At each time $k$, the watermark signal is chosen to be
\begin{align}
\phi_k = U_k^{1/2}\zeta_k ,
  \label{eq:phizeta}
\end{align}
where $\zeta_k$s are i.i.d. Gaussian random vectors with covariance $I$.

\subsection*{Inference on $H_\tau$}

The rest of this section is devoted to inferencing the system parameters from the collected sensory data $y_0,\ldots,y_k$ and watermarks $\phi_0,\ldots,\phi_k$. We will first identify the Markov parameters $H_\tau$ of the system.

Let us define the following quantity $H_{k,\tau}$, where $0\leq \tau\leq 3\tilde n-2$, as  
\begin{align}
  H_{k,\tau} &\triangleq \frac{1}{k-\tau+1} \sum_{t=\tau}^k y_t \phi_{t-\tau}^T U_{t-\tau}^{-1}\nonumber\\
  & = H_{k-1,\tau} + \frac{1}{k-\tau+1}\left(y_k \phi_{k-\tau}^T U_{k-\tau}^{-1} - H_{k-1,\tau}\right),\label{eq:Htauest}
\end{align}
\revised{where $H_{k,\tau}$ is an estimate of $H_{\tau}$ at time $k$.}

\begin{remark}
  It is worth noticing that other methods, such as subspace identification, may be superior for classical system identification tasks to the method we proposed. However, since the covariance of our watermark signal converges to a degenerate matrix (of rank $1$), it is non-trivial to analyze the convergence properties for more advanced system identification methods, such as subspace identification, which we shall leave as a further research direction.
\end{remark}

It is worth noticing that the calculation of the matrices $\mathcal U,\,\mathcal W,\,\mathcal P$ and $\mathcal X$ requires $H_\tau$ for all $\tau \geq 0$. Next we shall show that in fact only finitely many $H_\tau$s are needed to compute those matrices, which requires one intermediate result:

\begin{lemma}
  Assuming the matrix $A$ is diagonalizable with $\lambda_1,\ldots,\lambda_{\tilde n}$ being its distinct eigenvalues, then there exist unique $\Omega_1,\cdots,\Omega_{\tilde n}$, such that
  \begin{align}
    H_\tau = \sum_{i=1}^{\tilde n} \lambda_i^\tau \Omega_i.
  \end{align}
  \label{lemma:finitetoinf}
\end{lemma}
\begin{proof}
  Without loss of generality, we assume that $A$ is a diagonal matrix. As a result,
  \begin{align*}
    A^{\tau} = \diag(\lambda_1^{\tau}I_1, \lambda_2^{\tau}I_2, \cdots, \lambda_{\tilde n}^{\tau}I_{\tilde n}) , 
  \end{align*}
  where $\lambda_i$ denotes the $i$th distinct eigenvalue of $A$, \revised{$\lambda_i^{\tau}$ denotes $\lambda_i$ to the power of $\tau$}, and $I_i$ is the identity matrix of size $n_i$ by $n_i$ with $n_i$ the multiplicity of $\lambda_i$. Hence, we have
  \begin{align*}
    H_\tau &= CA^\tau B = \sum_{i=1}^{\tilde n} \lambda_i^\tau\Omega_i,
  \end{align*}
  with $\Omega_i = C\diag(0,\ldots,0,I_i,0,\ldots,0)B$, which completes the proof.
\end{proof}
Since $A$ satisfies its own minimal polynomial $p(x) = \prod_{i=1}^{\tilde n}(x-\lambda_i)=x^{\tilde n} + \alpha_{\tilde n-1}x^{\tilde n-1}+\ldots+\alpha_0$, we know that for any $i\geq 0$:
\begin{align}
  H_{i+\tilde n}+\alpha_{\tilde n-1}H_{i+\tilde n-1}+\cdots+\alpha_0H_i = CA^ip(A)B = 0.
  \label{eq:hminimalpoly}
\end{align}
Leveraging \eqref{eq:hminimalpoly}, we could use $H_0, H_1, \cdots,H_{3\tilde n-2}$ to estimate both $\lambda_i$s and $\Omega_i$s and thus $H_\tau$ for any $\tau$. To this end, let us define:
\begin{align}
  \label{eq:optcharpoly}
  \begin{bmatrix}
      \alpha_{k,0}\\
      \vdots\\
      \alpha_{k,\tilde n-1}
    \end{bmatrix} \triangleq -\Xi_k^{-1}\begin{bmatrix}
\tr(\mathcal H_{k,0}^T\mathcal H_{k,\tilde n})\\
\vdots\\
\tr(\mathcal H_{k,\tilde n-1}^T\mathcal H_{k,\tilde n})
    \end{bmatrix},
\end{align}
where
\begin{align*}
 \Xi_k \triangleq \begin{bmatrix}
\tr(\mathcal H_{k,0}^T\mathcal H_{k,0}) & \cdots & \tr(\mathcal H_{k,0}^T\mathcal H_{k,\tilde n-1}) \\
\vdots&\ddots&\vdots\\
\tr(\mathcal H_{k,\tilde n-1}^T\mathcal H_{k,0}) & \cdots & \tr(\mathcal H_{k,\tilde n-1}^T\mathcal H_{k,\tilde n-1}) \\
    \end{bmatrix},
\end{align*}
and
\begin{align*}
  \mathcal H_{k,i}\triangleq\begin{bmatrix}
    H_{k,i}\\
    H_{k,i+1}\\
    \vdots \\
    H_{k,i+2\tilde n-2}\\
  \end{bmatrix}.
\end{align*}

\begin{remark}
  One can prove that $\alpha_{k,i}$ from \eqref{eq:optcharpoly} is the solution of the following minimization problem:
  \begin{align*}
 \min \|\mathcal H_{k,\tilde n}+\alpha_{\tilde n-1}\mathcal H_{k,\tilde n-1}+\cdots+\alpha_0\mathcal H_{k,0}\|_F,
  \end{align*}
  where $\|\cdot\|_F$ denotes the Frobenius norm of a matrix.
\end{remark}

Let us denote the roots of the polynomial $p_k(x) = x^{\tilde n}+\alpha_{k,{\tilde n}-1}x^{{\tilde n}-1}+\cdots+\alpha_{k,0}$ to be $\lambda_{k,1},\cdots,\lambda_{k,{\tilde n}}$. Define a Vandermonde like matrix $V_k$ to be
\begin{align*}
  V_k\triangleq \begin{bmatrix}
    1 & 1 &\cdots&1\\
    \lambda_{k,1} & \lambda_{k,2} &\cdots&\lambda_{k,{\tilde n}}\\
    \vdots & \vdots &\ddots&\vdots\\
    \lambda^{3\tilde n-2}_{k,1} & \lambda^{3\tilde n-2}_{k,2} &\cdots&\lambda^{3\tilde n-2}_{k,\tilde n}\\
  \end{bmatrix},
\end{align*}
\revised{where $\lambda_{k,i}$ is an estimate of $\lambda_{i}$ at time $k$ and $\lambda_{k,i}^\tau$ is $\lambda_{k,i}$ to the power of $\tau$},
and we shall estimate $\Omega_i$ as
\begin{align}
  \begin{bmatrix}
    \Omega_{k,1}\\
    \vdots\\
    \Omega_{k,\tilde n}
  \end{bmatrix} =  \left(V_k\otimes I_m \right)^+\begin{bmatrix}
      H_{k,0}\\
      \cdots\\
      H_{k,3\tilde n-2}
    \end{bmatrix}.\label{eq:Omegaest}
\end{align}

\subsection*{Inference on $\varphi_k$, $\vartheta_k$ and $\mathcal W$}

This subsection is devoted to the inference of $\varphi_k$ and $\vartheta_k$ defined in \eqref{eq:varphidef}, which corresponds to the parts of $y_k$ generated by the watermark signal and noise respectively. We will further infer the covariance $\mathcal W$ of $\vartheta_k$.

Let us define
\begin{align}
\hat \varphi_k \triangleq \sum_{i=1}^{\tilde n} \hat \varphi_{k,i},\label{eq:varphiest}
\end{align}
with $\hat \varphi_{k,i} = \lambda_{k,i }\hat \varphi_{k-1,i} + \Omega_{k,i} \phi_{k},$ and $\hat \varphi_{-1,i}=0$. As a result, we can estimate $\vartheta_k$ as
\begin{align}
\hat \vartheta_k \triangleq y_k - \hat \varphi_k.\label{eq:varthetaest}
\end{align} The covariance of $\vartheta_k$ can be estimated as
\begin{align}
\mathcal W_k \triangleq \frac{1}{k+1}\sum_{t=0}^k \hat \vartheta_t \hat \vartheta_t^T.\label{eq:West}
\end{align}

\subsection*{Inference on $\mathcal P$, $\mathcal X$, $\mathcal U$ and $g_k$}
Finally we can derive an estimation of the $\mathcal P$ and $\mathcal X$ matrices, which are required to compute the optimal covariance $U$ of the watermark signal, given by
\revised{
\begin{align}
  \mathcal P_k &=\sum_{\tau=0}^\infty\left(\sum_{i=1}^{\tilde n}\lambda_{k,i}^\tau \Omega_{k,i}\right)^T\mathcal W_k^{-1} \left(\sum_{i=1}^{\tilde n}\lambda_{k,i}^\tau \Omega_{k,i}\right) \nonumber\\
  =&\sum_{\tau=0}^\infty\left( \sum_{i=1}^{\tilde n}\sum_{j=1}^{\tilde n} \lambda_{k,i}^\tau\lambda_{k,j}^\tau \Omega_{k,i}^T \mathcal W_k^{-1}\Omega_{k,j}\right)\nonumber\\
  =& \sum_{i=1}^{\tilde n}\sum_{j=1}^{\tilde n}  \left(\sum_{\tau=0}^\infty \left(\lambda_{k,i}\lambda_{k,j}\right)^\tau \right)   \Omega_{k,i}^T \mathcal W_k^{-1}\Omega_{k,j}\nonumber\\
  =& \sum_{i=1}^{\tilde n}\sum_{j=1}^{\tilde n} \frac{1}{1-\lambda_{k,i}\lambda_{k,j}} \Omega_{k,i}^T \mathcal W_k^{-1}\Omega_{k,j}, \label{eq:Pest}
\end{align}
where (28) is derived from the summation of geometric series,}
and
\begin{align}
  \mathcal X_k ={}& \sum_{\tau=0}^\infty\left(\sum_{i=1}^{\tilde n}\lambda_{k,i}^\tau \Omega_{k,i}\right)^TX_{yy} \left(\sum_{i=1}^{\tilde n}\lambda_{k,i}^\tau \Omega_{k,i}\right)\nonumber\\
                  &+ \sum_{i=1}^{\tilde n} \Omega_{k,i}^T X_{y\phi}+ X_{\phi y} \sum_{i=1}^{\tilde n} \Omega_{k,i}+X_{\phi\phi}\nonumber\\
  ={}& \sum_{i=1}^{\tilde n}\sum_{j=1}^{\tilde n} \frac{1}{1-\lambda_{k,i}\lambda_{k,j}} \Omega_{k,i}^T X_{yy} \Omega_{k,j} + \sum_{i=1}^{\tilde n} \Omega_{k,i}^T X_{y\phi}\nonumber\\
                  &+ X_{\phi y} \sum_{i=1}^{\tilde n} \Omega_{k,i}+X_{\phi\phi}.\label{eq:Xest}
\end{align}

The Neyman-Pearson detection statistics $g_k$ can be approximated by
\begin{align}\label{eq:gest}
  \hat g_k = & \left(y_{k}- \hat{\varphi}_{k}\right)^{{T}}  \mathcal W_k^{-1}\left(y_{k}- \hat{\varphi}_{k}\right)- y_{k}^{{T}}  \left(\mathcal W_k+\mathcal U_k\right)^{-1}y_{k},
\end{align}
where
\begin{align}
  \mathcal U_k &=\sum_{\tau=0}^\infty\left(\sum_{i=1}^{\tilde n}\lambda_{k,i}^\tau \Omega_{k,i}\right)U_{k,*} \left(\sum_{i=1}^{\tilde n}\lambda_{k,i}^\tau \Omega_{k,i}\right)^T \nonumber\\
               &= \sum_{i=1}^{\tilde n}\sum_{j=1}^{\tilde n} \frac{1}{1-\lambda_{k,i}\lambda_{k,j}} \Omega_{k,i} U_{k,*}\Omega_{k,j}^T. \label{eq:mathcalUest}
\end{align}
\revised{
\begin{remark}
	For the proposed online algorithm, the system identification and watermark design are tightly coupled. As is commented in Remark 10, the watermarking-based replay attack detection requires the injection of a rank-1 watermarking signal (assuming it is performed optimally). On the other hand, persistency of excitation is required for system identification, i.e., the injected signal needs to be full rank. As a result, we carefully design the covariance of the injected signal to be the ``optimal'' rank-1 covariance matrix on the current knowledge of the system, plus a diminishing factor $(k+1)^{-\beta}I$, and we further prove in this paper that this additional term, although vanishing asymptotically, provides us with enough information to perfectly identify the necessary parameters of the system. 
\end{remark}
\begin{remark}
	It is worth noticing that comparing to an approach with off-line system identification first and then watermarking design later, our approach provides the following advantages:
	\begin{enumerate}
		\item Theoretically speaking, finite-time system identification cannot identify the system parameters precisely and hence the watermarking scheme will not be optimal if the system identification process is stopped.
		\item In practice, the control system could slowly change due to various reasons (e.g., components wear out), so we need to adjust the parameters continuously.
		\item Moreover, for many practical control systems, a model of the system is not available. It is often too expensive to stop the system operation and to perform off-line system identification.
	\end{enumerate}
	We would like to further point out that the classical system identification procedure can be easily integrated to our approach, by providing better estimation of $\mathcal P$ and $\mathcal X$ in the initialization step in Algorithm 1. Hence, classical system identification approach complements our algorithm very well.
\end{remark}
}
\subsection{Algorithm Properties}
The following theorem establishes the convergence of $U_{k,*}$ and $g_k$, the proof of which is reported in the appendix for the sake of legibility.
\begin{theorem}\label{theorem:asconvergencerate}
Assuming that $A$ is strictly stable and Assumption 2 holds. If $0<\beta<1$, then for any $\epsilon > 0$, the following limits hold almost surely:
 \begin{align}
\lim_{k\rightarrow\infty}\frac{U_{k,*}-U_*}{k^{-\gamma+\epsilon}} = 0, \lim_{k\rightarrow\infty}\frac{\hat g_k-g_k}{k^{-\gamma+\epsilon}} = 0,
   \label{eq:asconvergencerate}
 \end{align}
where $\gamma =  (1-\beta)/2 > 0$. In particular, $U_{k,*}$ and $\hat g_k$ almost surely converge to $U_*$ and $g_k$ respectively.
\end{theorem}

From the definition of $U_k = U_{k,*} + (k+1)^{-\beta}\delta I $, we immediately have the following corollary:
\begin{corollary}
  \label{corollary:convergencerateofU}
  Assuming that $A$ is strictly stable and Assumption 2 holds. If $0<\beta<1$, then for any $\epsilon > 0$, the following limit holds almost surely:
  \begin{align}
    \lim_{k\rightarrow\infty}\frac{U_{k}-U_*}{k^{-\min(\gamma,\beta)+\epsilon}} = 0.\label{eq:asconvergencerate2}
  \end{align}
\end{corollary}
\begin{remark}
  It is worth noticing that \eqref{eq:asconvergencerate} implies that both $U_{k,*}-U_*$ and $\hat g_k-g_k$ are of the order $O(k^{-\gamma+\epsilon})$ as $k$ goes to infinity. Hence, the convergence rate $\gamma$ is maximized when $\beta \rightarrow 0^+$, which corresponds to the case where the exploration term $(k+1)^{-\beta}\delta I$ in $U_k$ stays constant. However, although this will maximize the performance for the inference algorithm, the covariance $U_k$ of the watermark signal $\phi_k$ will not converge to the true optimal $U_*$. {In order to achieve ``fastest'' convergence rate of $U_k$, we need to choose the decay rate for the exploration term to be $\beta = 1/3=\argmax_\beta(\gamma,\beta)$}.

  We would also like to point out that Theorem~\ref{theorem:asconvergencerate} only provides an upper bound for the almost sure convergence rate and we plan to investigate the exact convergence rate in our future work. It is also interesting to see if faster convergence can be achieved by using more advanced system identification techniques.
\end{remark}

\section{Simulation Result}
\label{sec:simulation}
In this section, the performance of the proposed algorithm is evaluated. We will apply the proposed online ``learning'' approach to a numerical example and an industrial process, Tennessee Eastman Process (TEP).

\subsection{A Numerical Example}
First we choose $m=3, n=5, p=2$ and $A,\,B,\,C$ are all randomly generated, with $A$ being stable. It is assumed that $X$ in \eqref{eq3}, the covariance matrices $Q$ and $R$ are all identity matrices with proper dimensions. We assume that $\delta$ in \eqref{eq:opt2} is equal to $10\%$ of optimal LQG cost $J_0$. Fig. \ref{fig:errustar} shows relative error $\|U_{k,*}-U_*\|_F/\|U_*\|_F$ of the estimated $U_{k,*}$ v.s. time $k$ for different $\beta$s.

\begin{figure}[ht]
	\input{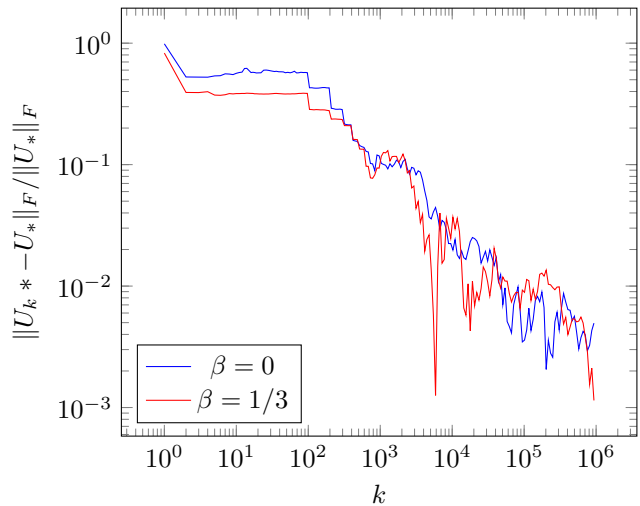} 
	\caption{\label{fig:errustar} Relative error of $U_{k,*}$ for different $\beta$. The red solid line denotes the relative error of $U_{k,*}$ when $\beta = 0$. The blue solid line is the relative error of $U_{k,*}$ when $\beta = 1/3$.}
\end{figure}

From Fig~\ref{fig:errustar}, one can see that the estimator error converges to $0$ as time $k$ goes to infinity and the convergence approximately follows a power law. From Theorem~\ref{theorem:asconvergencerate}, we know that $U_{k,*}-U_*\sim O(k^{-\gamma+\epsilon})$, where $\gamma = (1-\beta)/2$. However, from Fig~\ref{fig:errustar}, it seems that the convergence speed of the error for different $\beta$ is comparable. Notice that Theorem~\ref{theorem:asconvergencerate} only provides an upper bound for the convergence rate. As a result, it would be interesting to quantify the exact impact of $\beta$ on the convergence rate, which we shall leave as a future research direction.

Now we consider the detection performance of our online watermark signal design, after an initial inference period, where no attack is present. It is assumed that the attacker records the sensor readings from time $10^4+1$ to $10^4+100$ and replays them to the system from time $10^4+101$ to $10^4+200$. Fig~\ref{fig:learningvsnoraml} shows the trajectory of the Neyman-Pearson statistic $g_k$ and our estimate $\hat g_k$ of $g_k$ for one simulation. Notice that $\hat g_k$ can track $g_k$ with high accuracy. Furthermore, both $\hat g_k$ and $g_k$ are significantly larger when the system is under replay attack (after time $10^4+101$). Hence one can conclude that even without parameter knowledge, we can successfully estimate $g_k$ and detect the presence of the replay attack.

\begin{figure}[ht]
  \input{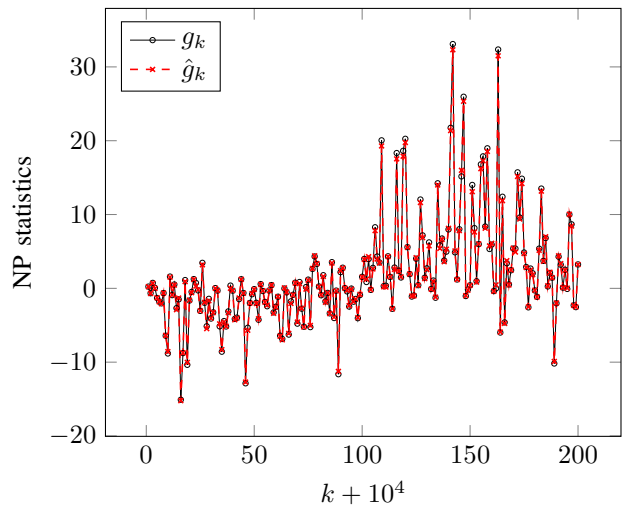}
  \caption{\label{fig:learningvsnoraml} The detection statistics v.s. time. The black solid line with circle markers is the true Neyman-Pearson statistics $g_k$, assuming full system knowledge. The red dashed line with cross markers denotes our estimated $\hat g_k$.}
\end{figure}

\subsection{TEP Example}
Tennessee Eastman Process (TEP) is a commonly used process control system proposed by Downs and Vogel in~\cite{downs1993plant}. In this simulation, we adopt a simplified version of TEP from~\cite{ricker1993model}, as follows:
\begin{align*}
\dot{x} &= Ax+ Bu,\\
y &= Cx,
\end{align*}
where $A,B$ and $C$ are constant matrices~\footnote{For more details about this dynamic model, please refer to Appendix \Rmnum{1} in~\cite{ricker1993model}.}.  

This system simulates a MIMO system of order $n = 8$ with $p = 4$ inputs and $m = 10$ outputs. We discretize the system using the control system toolbox in MATLAB, by selecting a sample time of $0.6s$. Again, we choose $X$ in \eqref{eq3}, the covariance matrices $Q$ and $R$ to be identity matrices with proper dimensions. We assume that $\delta$ in \eqref{eq:opt2} is equal to $5\%$ of $J_0$, and $\beta = 1/3$. In this simulation, we assume that we do not know the dimension of the state space, which is $8$, and instead we underestimate it by assuming that $A$ only has $\tilde n=5$ distinct eigenvalues.

Fig. \ref{fig:errustarte} illustrates the relative error $\|U_{k,*}-U_*\|_F/\|U_*\|_F$ after running the system for roughly 1 week ($10^6\times 0.6s\approx 0.992$week). Fig~\ref{fig:learningvsnoraml_te} illustrates the NP statistics $g_k$ and the estimated NP statistics $\hat g_k$, assuming that the adversary collects the measurement from $10^6+1$ to $10^6+100$ and replays them to the system from time $10^6+101$ to $10^6+200$. One can see that although we underestimate the dimensions of the system, our algorithm can still achieve a high accuracy.

\begin{figure}[ht]
	\begin{tikzpicture}[scale=0.7]
\begin{axis}[ylabel = {$\|U_k*-U_*\|_F/\|U_*\|_F$}, ymode = {log}, xlabel = {$k$}, xmode = {log}]\addplot [mark = {none}]coordinates {
(1.0, 0.9166687883495798)
(2.0, 0.9722987246723181)
(3.0, 0.9959700306319388)
(4.0, 0.99619481055794)
(5.0, 0.9895372765540812)
(6.0, 0.9282081931256069)
(7.0, 0.9686001873785345)
(8.0, 0.9363243800493992)
(9.0, 0.919803272829615)
(10.0, 0.9124541427279479)
(11.0, 0.9246359584378264)
(12.0, 0.9139201400905902)
(13.0, 0.913191332402363)
(14.0, 0.9121341846861449)
(16.0, 0.9121543260046752)
(17.0, 0.9120393049493651)
(18.0, 0.911942471422965)
(19.0, 0.9119116828444889)
(21.0, 0.9127180051931093)
(22.0, 0.9126137290916178)
(24.0, 0.9141584244698505)
(26.0, 0.9137594851594167)
(28.0, 0.9129978495115476)
(30.0, 0.9133464424090457)
(32.0, 0.9129591180591482)
(34.0, 0.9131473356074395)
(36.0, 0.9136430496112449)
(39.0, 0.9141158997252693)
(42.0, 0.9163445514533404)
(45.0, 0.9151591937564608)
(48.0, 0.9156532054854291)
(52.0, 0.9150463498010641)
(56.0, 0.9150216696429136)
(60.0, 0.9151329813872797)
(64.0, 0.9148972918085316)
(69.0, 0.914493266511644)
(74.0, 0.9146085909777327)
(79.0, 0.9143509533941867)
(85.0, 0.9141973486262378)
(91.0, 0.9145209625575896)
(97.0, 0.9140279375630509)
(104.0, 0.6056889642061132)
(112.0, 0.6056424341802338)
(120.0, 0.6056422043351312)
(128.0, 0.6059338630931108)
(138.0, 0.6059758348638881)
(148.0, 0.6061780251190082)
(158.0, 0.6064238490519104)
(170.0, 0.607008372617737)
(182.0, 0.6070617577691271)
(195.0, 0.607333564072601)
(209.0, 0.6268485137140902)
(224.0, 0.6266214335477034)
(240.0, 0.6262051063127955)
(258.0, 0.6254212274120868)
(276.0, 0.6252206000278123)
(296.0, 0.6248621927834054)
(318.0, 0.7192660886068187)
(340.0, 0.7131095789541924)
(365.0, 0.7121846209874493)
(391.0, 0.7100903859109878)
(419.0, 0.6390577819903255)
(450.0, 0.6311765808391364)
(482.0, 0.6273939294729601)
(517.0, 0.6774647085525551)
(554.0, 0.6664985545573172)
(594.0, 0.6521201904809307)
(636.0, 0.5501832354493805)
(682.0, 0.555490230580253)
(731.0, 0.5841407221441295)
(784.0, 0.5727420809933143)
(840.0, 0.5588350491780484)
(901.0, 0.49208547534995006)
(965.0, 0.48468712002039527)
(1035.0, 0.7824407371521901)
(1109.0, 0.4682686389710935)
(1189.0, 0.4766242993060242)
(1275.0, 0.4553421347079951)
(1366.0, 0.4775131077750679)
(1464.0, 0.5142913446786829)
(1570.0, 0.44730965666917893)
(1683.0, 0.6227347440459559)
(1804.0, 0.43156765422382676)
(1933.0, 0.42604268360232345)
(2072.0, 0.36874722561837486)
(2221.0, 0.3666751038288462)
(2381.0, 0.3188233715700378)
(2552.0, 0.5597138253971148)
(2736.0, 0.6527319442046979)
(2933.0, 0.532237214011298)
(3144.0, 0.5082476007939787)
(3370.0, 0.5738616343112553)
(3612.0, 0.3415894084737519)
(3872.0, 0.3025645895540138)
(4150.0, 0.25301330609942785)
(4448.0, 0.19643505662191846)
(4768.0, 0.22000696093945313)
(5111.0, 0.22277786444372089)
(5478.0, 0.24915672111572457)
(5872.0, 0.278457116768585)
(6294.0, 0.2253451745118907)
(6747.0, 0.2026340152105261)
(7232.0, 0.27585772928234814)
(7752.0, 0.25865092515740684)
(8309.0, 0.2293476302234855)
(8907.0, 0.16708549688223207)
(9547.0, 0.2293060945958305)
(10234.0, 0.2118997961266888)
(10969.0, 0.19475069882238255)
(11758.0, 0.19365391749208866)
(12603.0, 0.16095920702933075)
(13509.0, 0.16628116702090978)
(14481.0, 0.1672330271749415)
(15522.0, 0.1887453240595656)
(16638.0, 0.16446942438702067)
(17834.0, 0.13674941357798898)
(19116.0, 0.15761917775849885)
(20490.0, 0.11199341893977073)
(21963.0, 0.10799388718530953)
(23542.0, 0.0724175971692653)
(25235.0, 0.03476258634043264)
(27049.0, 0.08058149220318203)
(28994.0, 0.011042430954321199)
(31078.0, 0.012548258887310634)
(33312.0, 0.033540663053581846)
(35707.0, 0.06471466089693292)
(38274.0, 0.08223751140885001)
(41026.0, 0.10156744871368978)
(43976.0, 0.12279609867677929)
(47137.0, 0.12443415060684421)
(50526.0, 0.09699669570082703)
(54158.0, 0.07177574359447368)
(58052.0, 0.04358694464814466)
(62225.0, 0.051812096445747524)
(66699.0, 0.09309831984888123)
(71494.0, 0.056190815649890646)
(76634.0, 0.04261708185745748)
(82143.0, 0.05719274083975746)
(88048.0, 0.044666264078626096)
(94378.0, 0.008208391496132015)
(101163.0, 0.02268004702027547)
(108436.0, 0.02648088489652857)
(116232.0, 0.007564006275972799)
(124588.0, 0.015047490962574645)
(133545.0, 0.035970689212181425)
(143145.0, 0.0447065829019502)
(153436.0, 0.03982092589748738)
(164467.0, 0.020343599387374532)
(176291.0, 0.03189348128262503)
(188965.0, 0.03242753065219538)
(202550.0, 0.0254580591916076)
(217111.0, 0.028543887039771308)
(232720.0, 0.014890594910713068)
(249450.0, 0.025512108980704162)
(267384.0, 0.01751094857758149)
(286606.0, 0.021444553275889875)
(307211.0, 0.02315424251646576)
(329297.0, 0.02132614933340774)
(352970.0, 0.014556151037797754)
(378346.0, 0.014110341632664409)
(405546.0, 0.014228000421777015)
(434701.0, 0.01902599596417682)
(465952.0, 0.02254494463259949)
(499450.0, 0.00961477140343008)
(535356.0, 0.02693971972042018)
(573844.0, 0.028600069494930845)
(615098.0, 0.04350321641372623)
(659318.0, 0.03694440388771107)
(706718.0, 0.020794180615125786)
(757525.0, 0.019402780289807842)
(811984.0, 0.019249370307906157)
(870359.0, 0.014341837922416651)
(932930.0, 0.012674822727721701)
};
\end{axis}

\end{tikzpicture} 
	\caption{Relative error of $U_{k,*}$.}
	\label{fig:errustarte}
\end{figure}
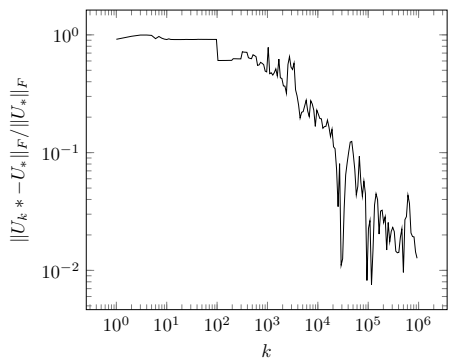

\begin{figure}[ht]
	\input{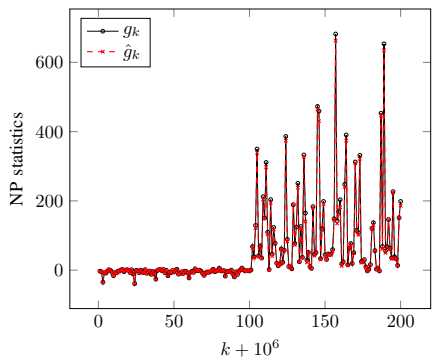} 
  \caption{The detection statistics v.s. time. The black solid line with circle markers is the true Neyman-Pearson statistics $g_k$, assuming full system knowledge. The red dashed line with cross markers denotes our estimated $\hat g_k$.}
\label{fig:learningvsnoraml_te} 
\end{figure}

\section{Conclusion}
\label{sec:conclusion}
In this paper, an algorithm that can simultaneously generate the watermarking signal and infer the system parameters is proposed. We prove that our algorithm converges to the optimal one and characterize an upper bound for the almost surely convergence rate. For future works, we would like to quantify the exact convergence rate, as well as exploring other system identification methods and prove their convergence. \revised{It is of interest to study secure control in other cases, such as batach-operating process. We are also interested in adversarial learning when the sensor data is compromised.}

\appendices
\section{Proof of Theorem~\ref{theorem:asconvergencerate}}

The whole appendix is devoted to proving Theorem~\ref{theorem:asconvergencerate}. We shall present several preliminary results first and then proceed with the proof of Theorem~\ref{theorem:asconvergencerate}.
\subsection*{Preliminary Results}
 To simplify notations, for a random variable (vector, or matrices) $x_k$, we denote that $x_k\sim \mathcal C(\alpha)$ if for all $\epsilon > 0$, we have that $x_k \sim O(k^{\alpha+\epsilon})$, i.e.,
\begin{align*}
  \lim_{k\rightarrow \infty} \frac{\|x_k\|}{k^{\alpha + \epsilon}}\asequal 0.
\end{align*}

Notice that $x_k \sim O(k^\alpha)$ implies that $x_k\sim \mathcal C(\alpha)$, but the reverse is not  necessarily true\footnote{To see a counterexample, $\log k \sim \mathcal C(0)$, but $\log k$ is not of the order $O(k^0)$.}. The following lemma establishes some basic properties of $\mathcal C(\alpha)$ functions:

\begin{lemma}
  \label{lemma:functionclass}
  Assuming that $x_k\sim \mathcal C(\alpha)$ and $y_k\sim \mathcal C(\beta)$, with $\alpha \geq \beta$, then the following statements hold:
  \begin{enumerate}
  \item $x_k + y_k \sim \mathcal C(\alpha)$, $x_k\times y_k \sim \mathcal C(\alpha+\beta)$, and $(x_k+\Delta x_k)(y_k+\Delta y_k)-x_ky_k \sim \mathcal C(\max\{\alpha\beta',\alpha'\beta,\alpha'\beta'\})$, suppose that $\Delta x_k\sim \mathcal C(\alpha')$ and $\Delta y_k\sim\mathcal C(\beta')$.
  \item $\sum_{t=0}^k x_t \sim  \mathcal C(\alpha+1)$.
  \item Suppose $f$ is differentiable at $0$ and $\alpha < 0$, then
\begin{align*}
    f(x_k)-f(0)\sim \mathcal C(\alpha).
\end{align*}
  \item $s_k \sim   \mathcal C(\alpha)$, with
    \begin{align*}
s_k = \rho s_{k-1}+x_k,\,s_{-1} = 0,
    \end{align*}
    where $|\rho|<1$.
  \item  Assume that $X_k$ is a matrix and $X_k \sim \mathcal C(\alpha)$. Let
    \begin{align*}
      S_k = A S_{k-1}B+X_k,\,S_{-1} = 0,
    \end{align*}
    where $A, B$ are matrices of proper dimensions. Then $S_k\sim \mathcal C(\alpha)$ if $B^T\otimes A$ is strictly stable.
    \item  $\zeta_k\sim \mathcal C(0)$, where $\{\zeta_k\}$ is a sequence of i.i.d. Gaussian random variables, i.e., $\zeta_k\sim \mathcal N(\bar \mu,Z)$. 
  \end{enumerate}
\end{lemma}

\begin{proof}
  The first three statements can be trivially proved and hence we only focus on the last three statements.
  \begin{enumerate}
\setcounter{enumi}{3} 
  \item Since $x_k\sim \mathcal C(\alpha)$, it is easy to see that for any $\epsilon > 0$,
  \begin{align*}
   \sup_k \frac{|x_k|}{k^{\alpha+\epsilon}} = M_a(\epsilon) < \infty. \quad a.s.
  \end{align*}
  As a result,
  \begin{align*}
\frac{| s_k|}{k^{\alpha+2\epsilon}} \leq \frac{1}{k^\epsilon} \sum_{i=1}^k\left|\rho^{i-1}\right|  \left|\frac{x_{k-i}}{k^{\alpha+\epsilon}}\right| \leq  \frac{1}{k^\epsilon} \frac{M_a(\epsilon)}{1-|\rho|},
  \end{align*}
  which almost surely converges to $0$ as $k$ goes to $\infty$. As a result, $s_k\sim \mathcal C(\alpha)$.
  \item For the last statement, notice that
\begin{align*}
  vec(S_k) = \left(B^T\otimes A\right)  vec( S_{k-1})+vec(X_k).
\end{align*}
Therefore, the argument that $S_k\sim\mathcal C(\alpha)$ follows the same line of proof as the fourth statement.
\item We only need to prove for the case where $\zeta_k$ follows the standard normal distribution. The high dimensional case can then be proved by checking each entry of $\zeta_k$ with proper scaling and shifting. For any $\epsilon,\phi > 0$, we have
  \begin{align*}
    P\left(\frac{|\zeta_k|}{k^\epsilon} >\phi \right) &=\sqrt{\frac{2}{\pi}}\int_{\phi k^\epsilon} ^\infty \exp(-x^2/2)dx.
  \end{align*}
  Suppose that $k$ is large enough, such that $\phi k^\epsilon > 1$, then we have
  \begin{align*}
    \int_{\phi k^\epsilon} ^\infty \exp(-x^2/2)dx &\leq \int_{\phi k^\epsilon} ^\infty \exp(-x^2/2)\times xdx\\
                                                  &= \exp(-\phi^2 k^{2\epsilon} /2),
  \end{align*}
  and
  \begin{align*}
    \lim_{k\rightarrow\infty}k^2 \exp(-\phi^2 k^{2\epsilon} /2) = 
    \lim_{x\rightarrow\infty}\left(\frac{2x}{\phi^2}\right)^{1/\epsilon} \exp(-x) = 0.
  \end{align*}
  As a result, using direct comparison test for infinite series, we can prove that
  \begin{align*}
    \sum_{k=1}^\infty  P\left(\frac{|\zeta_k|}{k^\epsilon} >\phi \right) < \infty,
  \end{align*}
  which further implies (by Borel-Cantelli Lemma), that
  \begin{align*}
    \limsup_{k\rightarrow\infty} \frac{|\zeta_k|}{k^\epsilon} \leq \phi,\,a.s.
  \end{align*}
  Since $\phi$ can be arbitrarily small, $\zeta_k/k^\epsilon\rightarrow 0$ almost surely, which finishes the proof.
  \end{enumerate}
\end{proof}

 Let $\{\mathcal F_k\}$ be a filtration of sigma algebras and $\{M_k\}$ be a matrix-valued stochastic process that is adapted to the filtration $\{\mathcal F_k\}$, we call $\{M_k\}$ a (matrix-valued) matingale (with respect to the filtration $\{\mathcal F_k\}$) if the following equality holds for all $t$:
\begin{align*}
  \mathbb E\left(M_{k+1}|\mathcal F_k\right) =  M_k.
\end{align*}

For the rest of the paper, we shall assume that the filtration $\mathcal F_k$ is the $\sigma$-algebra generated by the random variables $\{x_{-1}, \phi_0,\cdots, \phi_{k}, w_0,\cdots, w_{k},v_0,\cdots,v_k\}$. Now we have the following lemma to establish a strong law for matrix-valued martingale:
\begin{lemma}
  \label{lemma:MatrixStrongLaw}
  If $M_{k}=\varPhi_0+\varPhi_1+\cdots+\varPhi_k$ is a matrix-valued martingale such that 
  \begin{align*}
    \mathbb E\left\|\varPhi_k\right\|^2 \sim \mathcal C(\beta),
  \end{align*}
  where $0 \leq \beta < 1$, then $M_k/k$ converges to $0$ almost surely. Furthermore,
  \begin{align*}
    \frac{M_k}{k}\sim \mathcal C\left(\frac{\beta-1}{2}\right).
  \end{align*}
\end{lemma}

\begin{proof}
  Let $\varPhi_{k,ij}$ ($M_{k,ij}$) be the $(i,j)$-th entry of the matrix $\varPhi_k$ ($M_{k}$). It is easy to prove that $\{M_{k,ij}\}$ is a scalar martingale (adapted the same filtration $\{\mathcal F_k\}$) and since\footnote{This is due the fact that $\|A\| = \sup_{\|u\|=\|v\|=1} |u^TAv| \geq |e_i^TAe_j|$.}
  \begin{align*}
    \varPhi_{k,ij}^2\leq \|\varPhi_k\|^2,
  \end{align*}
  we have that $\mathbb E \varPhi_{k,ij}^2 \sim \mathcal C(\beta)$. For simplicity, let us define $\kappa = (\beta+1)/2$. One can easily verify that for any $\epsilon > 0$ and large enough $i$, the following inequalities hold:
  \begin{align*}
    i^{1-1} \left(k^{-\kappa -\epsilon}\right)^{2-2}  = 1 ,
  \end{align*}
  and
  \begin{align*}
    \sum_{k=i}^\infty&\left(k^{-\kappa -\epsilon}\right)^2 k^{-1} \leq \int_{i-1}^\infty x^{-2\kappa-2\epsilon-1}dx\\
                                                                 & = \frac{1}{2\kappa+2\epsilon} \left(i-1\right)^{-2\kappa-2\epsilon} \leq \frac{1}{\kappa} \left(i^{-\kappa-\epsilon}\right)^2.
  \end{align*}
  The last inequality is true since $\epsilon > 0$ and for large enough $i$, $(i-1)/i \rightarrow 1$.

  Finally, one can prove the following equality
  \begin{align*}
    \left( k^{-\kappa-\epsilon}\right)^2 \mathbb E \varPhi_{k,ij}^2 
    &=   k^{-\beta -2\epsilon-1}\mathbb E \varPhi_{k,ij}^2 \\
    & =   k^{-1-\epsilon}\frac{\mathbb E \varPhi_{k,ij}^2 }{k^{\beta+\epsilon}} \sim O(k^{-1-\epsilon}),
  \end{align*}
  which implies that
  \begin{align*}
    \sum_{k=1}^\infty\left( k^{-\kappa-\epsilon}\right)^2 \mathbb E \varPhi_{k,ij}^2 < \infty.
  \end{align*}

  As a result, by Lemma~1 in \cite{Chow1967}, we can deduce that
  \begin{align}
    \lim_{k\rightarrow\infty} \frac{M_{k,ij}/k}{k^{\kappa-1+\epsilon}} = \lim_{k\rightarrow\infty}k^{-\kappa-\epsilon} M_{k,ij} \asequal 0.
    \label{eq:entryconvergence}
  \end{align}
  Notice that \eqref{eq:entryconvergence} is true for all entries of the matrix $M_k$. Therefore, $M_k/k\sim\mathcal C(\kappa - 1)$, with $\kappa - 1 = (\beta +1)/2 - 1 = (\beta - 1)/2$. Since $\beta < 1$, $M_k/k$ converges to 0 almost surely.
\end{proof}

We are now ready to prove Theorem~\ref{theorem:asconvergencerate}, which requires several intermediate steps. 

\subsection*{Boundedness of $U_k$}

\begin{lemma}
  $U_k$ is upper and lower bounded by:
  \begin{align}
  \delta (k+1)^{-\beta}I\leq U_k \leq \delta\left( \left(X_{\phi\phi}-X_{\phi y}X_{yy}^{-1}X_{y\phi}\right)^{-1} + I\right).
    \label{eq:Ubound}
  \end{align}
  \label{lemma:Ubound}
\end{lemma}
\begin{proof}
  The first inequality is trivially true since $U_k = U_{k,*} + \delta (k+1)^{-\beta}I$. For the second inequality, notice that
\begin{align*}
  \mathcal X_k &\geq\left(\sum_{i=1}^{\tilde n}\Omega_{k,i}\right)^TX_{yy} \left(\sum_{i=1}^{\tilde n}\Omega_{k,i}\right)+ \sum_{i=1}^{\tilde n} \Omega_{k,i}^T X_{y\phi}\\
               &+ X_{\phi y} \sum_{i=1}^{\tilde n} \Omega_{k,i}+X_{\phi\phi}\\
               &\geq X_{\phi\phi}-X_{\phi y}X_{yy}^{-1}X_{y\phi}.
\end{align*}
As a result, $\tr(U_{k,*} \mathcal X_k)\leq \delta$ implies that
\begin{align*}
U_{k,*}\leq \delta \mathcal X_k^{-1} = \delta \left(X_{\phi\phi}-X_{\phi y}X_{yy}^{-1}X_{y\phi}\right)^{-1},
\end{align*}
and
\begin{align*}
  U_k \leq  U_{k,*}+\delta I =\delta\left( \left(X_{\phi\phi}-X_{\phi y}X_{yy}^{-1}X_{y\phi}\right)^{-1} + I\right).
\end{align*}

\end{proof}
\subsection*{Convergence of $H_{k,\tau}$}
\begin{lemma}
  $H_{k,\tau} - H_\tau \sim \mathcal C(-\gamma)$, with $\gamma = (1-\beta)/2$. In particular, $H_{k,\tau}$ converges to $H_\tau$ almost surely.
\end{lemma}

\begin{proof}
  It is easy to see that $y_k$ and $U_{k+1}$ are measurable w.r.t. $\mathcal F_k$. Furthermore, let $k_1,k_2\geq 0$ be two time indices, then it is easy to prove that
  \begin{align}
    \mathbb E(\phi_{k_1} \phi_{k_2+1}^T|\mathcal F_{k_2}) &= \begin{cases}
      U_{k_2+1} & \text{if }k_1 = k_2 +1\\
      0&\text{otherwise}
    \end{cases},\nonumber\\
    \mathbb E(w_{k_1} \phi_{k_2+1}^T|\mathcal F_{k_2})&= 0,\, \mathbb E(v_{k_1} \phi_{k_2+1}^T|\mathcal F_{k_2})= 0,    \label{eq:expectationfiltration}
  \end{align}
  which, combined with \eqref{eq:yexpansion}, implies that
  \begin{align}
    \mathbb E\left(y_{k+\tau}\phi_{k}^TU_k^{-1}|\mathcal F_{k-1}\right) = H_\tau.
    \label{eq:yphicorelation}
  \end{align} 

  Next we shall compute the expectation of $\|y_{k+\tau}\phi_{k}^TU_{k}^{-1}\|^2$. Notice that from \eqref{eq:phizeta}, $\phi_k = U_k^{1/2}\zeta_k$, where $\zeta_k$ follows the standard normal distribution. Hence,
  \begin{align*}
    \|&y_{k+\tau}\phi_{k}^TU_{k}^{-1}\|^2 = \|y_{k+\tau}\phi_{k}^TU_{k}^{-2}\phi_ky_{k+\tau}^T\|\\
      &\leq \|y_{k+\tau}\|^2\|\zeta_k\|^2 \|U_k^{-1}\| \leq \delta (k+1)^\beta \|y_{k+\tau}\|^2\|\zeta_k\|^2 .
  \end{align*}
  The last inequality is true due to \eqref{eq:Ubound}. As a result, by Cauchy-Schwarz inequality,
  \begin{align*}
    \mathbb E \|y_{k+\tau}\phi_{k}^TU_{k}^{-1}\|^2 \leq \delta(k+1)^{\beta} \sqrt{\mathbb E\|y_{k+\tau}\|^4}\sqrt{\mathbb E\|\zeta_k\|^4}.
  \end{align*}

  Notice that $\|\zeta_k\|$ is $\chi$-distributed with $p$ degree of freedom, which implies that $\mathbb E\|\zeta_k\|^4 = p(p+2)$. On the other hand, one can prove that $\sup_k \mathbb E\|y_{k}\|^4$ is bounded since by \eqref{eq:Ubound}, $U_k$ is upper bounded. As a result, we prove that 
   \begin{align*}
\mathbb E \|y_{k+\tau}\phi_{k}^TU_{k}^{-1}\|^2\sim \mathcal C(\beta),
   \end{align*} 
   which further implies that
  \begin{align}
    \mathbb E&\|y_{k+\tau}\phi_{k}^TU_{k}^{-1}-H_\tau\|^2 \leq \mathbb E \left(\|y_{k+\tau}\phi_{k}^TU_{k}^{-1}\|+\|H_\tau\|\right)^2\nonumber\\
   &\leq \mathbb E \left(2 \|y_{k+\tau}\phi_{k}^TU_{k}^{-1}\|^2 + 2\|H_\tau\|^2\right)\sim \mathcal C(\beta).\label{eq:yphicovariance}
  \end{align} 
As a result, by \eqref{eq:yphicorelation}, one can prove that the following stochastic process is a matrix-valued martingale
  \begin{align}
    \mathcal S_{\tau,i}(k+1) = \mathcal S_{\tau,i}(k) + \left[y_{(k+1)\tilde\tau+i}\phi_{k\tilde\tau+i+1}U_{k\tilde\tau+i+1}^{-1}-H_{\tau}\right]
  \end{align}
  for the filtration $\mathcal F_{k\tilde\tau+i}$, where $\tilde\tau = \tau + 1$, and $0\leq i\leq\tau$. Now by \eqref{eq:yphicovariance} and Lemma~\ref{lemma:MatrixStrongLaw}, we know that
 \begin{align*}
   \frac{\mathcal S_{\tau,i}(k) }{k} \sim \mathcal C(-\gamma).
 \end{align*} 
 From the definition of $\mathcal S_{\tau,i}(k)$, one can see that for large enough $k$,
 \begin{align}
H_{k,\tau} -H_\tau = \sum_{i=0}^\tau \frac{k_i}{k}\times \frac{\mathcal S_{\tau,i}(k_i)}{k_i}.
 \end{align}
 where $k_i = \max\{t\in \mathbb N:t\tilde \tau+i \leq k\}$. Notice that $k_i \geq 0$ and $\sum k_i = k$. Hence, the estimation error of $H_{k,\tau}-H_\tau$ is a convex combination of $\mathcal S_{\tau,i}$s. As a result, for any $\epsilon > 0$
\begin{align}
 \|H_{k,\tau}-H_\tau \| \leq \max_{0\leq i\leq \tau}\frac{\|\mathcal S_{\tau,i}(k_i)\|}{k_i}\sim O(k_i^{-\gamma+\epsilon}),
  \label{eq:HandS}
\end{align}
Notice that when $k$ is large enough, $k/k_i\rightarrow \tau$, which implies that $H_{k,\tau}-H_\tau \sim \mathcal C(-\gamma).$ The a.s. convergence can be trivially proved by the fact that $\gamma>0$ is positive.
\end{proof}

\subsection*{Convergence of $\lambda_{k,i}$ and $\Omega_{k,i}$}

Notice that due to the convergence of $H_{k,\tau}$ to $H_\tau$, we have that $\Xi_k$ converges to $\Xi$, where
\begin{align*}
\Xi \triangleq    \begin{bmatrix}
      \tr(\mathcal H_{0}^T\mathcal H_{0}) & \cdots & \tr(\mathcal H_{0}^T\mathcal H_{\tilde n-1}) \\
      \vdots&\ddots&\vdots\\
      \tr(\mathcal H_{\tilde n-1}^T\mathcal H_{0}) & \cdots & \tr(\mathcal H_{\tilde n-1}^T\mathcal H_{\tilde n-1}) \\
    \end{bmatrix},
\end{align*}
with
\begin{align*}
  \mathcal H_{i}\triangleq\begin{bmatrix}
    H_{i}\\
    \vdots \\
    H_{i+2\tilde n-2}\\
  \end{bmatrix}.
\end{align*}

We shall first prove that $\Xi$ is invertible. Suppose that there exists $\tilde \alpha = \left[\tilde \alpha_0,\ldots,\tilde \alpha_{\tilde n-1}\right]^T$, such that $\Xi\tilde \alpha = 0$, then 
\begin{align*}
0= \tilde \alpha^T\Xi\tilde \alpha = \left\|\sum_{i=0}^{\tilde n-1}\mathcal H_i\tilde \alpha_i\right\|_F^2,
\end{align*}
which further implies that $CA^{i}\tilde p(A)B = 0$ for all $0\leq i\leq 2n-2$, and $\tilde p(x) = \tilde \alpha_{\tilde n-1}x^{\tilde n-1}+\cdots+\tilde \alpha_0$. Hence, we know that
\begin{align*}
 \begin{bmatrix}
   C\\
   \vdots\\
  CA^{\tilde n-1}
\end{bmatrix}\tilde p(A)\begin{bmatrix}
B&\cdots&A^{\tilde n-1}B
\end{bmatrix}=0.
\end{align*}
By the fact that $(A,B)$ is controllable and $(A,C)$ is observable, $\tilde p(A)$ must be $0$. However, since $p(x)$ is minimal polynomial of $A$, $\tilde p(x)$ must be constantly $0$, which proves that $\tilde \alpha = 0$ and $\Xi$ is invertible.

Let us denote $\alpha_i$s as the coefficients of the minimal polynomial $p(x) = x^{\tilde n}+\alpha_{\tilde n-1}x^{\tilde n-1}+\cdots+\alpha_0$ of $A$, i.e., the monic polynomial with minimum degree. As a result, we have
\begin{align}
H_{i+\tilde n}+\alpha_{\tilde n-1}H_{i+\tilde n-1}+\cdots+\alpha_0H_i = CA^ip(A)B = 0.
  \label{eq:minimalH}
\end{align}
Hence, one can prove that,
\begin{align}
 \begin{bmatrix}
      \alpha_{0}\\
      \vdots\\
      \alpha_{\tilde n-1}
    \end{bmatrix} =-\Xi^{-1}\begin{bmatrix}
      \tr(\mathcal H_{0}^T\mathcal H_{\tilde n})\\
      \vdots\\
      \tr(\mathcal H_{\tilde n-1}^T\mathcal H_{\tilde n})
    \end{bmatrix},
  \label{eq:alphapolynomial}
\end{align}
 which, combined with the fact that $H_{k,\tau}-H_\tau\sim \mathcal C(-\gamma)$ and Lemma~\ref{lemma:functionclass}.3, proves that $\alpha_{k,i}-\alpha_k\sim \mathcal C(-\gamma)$. Since all the roots of the polynomial $p(x)$ are distinct, we can prove (see \cite{polynomialroots}) that $\lambda_{k,i}$s are differentiable functions of $\alpha_{k,i}$s at a neighborhood of $\alpha_i$, which further proves that $\lambda_{k,i}-\lambda_i\sim \mathcal C(-\gamma)$.

 Now let us define $V$ to be
 \begin{align*}
  V\triangleq \begin{bmatrix}
    1 & 1 &\cdots&1\\
    \lambda_{1} & \lambda_{2} &\cdots&\lambda_{{\tilde n}}\\
    \vdots & \vdots &\ddots&\vdots\\
    \lambda^{3\tilde n-2}_{1} & \lambda^{3\tilde n-2}_{2} &\cdots&\lambda^{3\tilde n-2}_{n}\\
  \end{bmatrix}.
 \end{align*}
Since $p(x)$ is the minimal polynomial of $A$ and $A$ is diagonalizable, the roots $\lambda_i$s of $p(x)$ are distinct, which proves that $V$ is of full column rank, i.e., $\rank(V) = \tilde n$. Therefore,
\begin{align*}
  \rank(V\otimes I_m) = \rank(V)\times \rank(I_m) = \tilde nm,
\end{align*}
which implies that $V\otimes I_m$ is of full column rank. 

Therefore, by Lemma~\ref{lemma:finitetoinf},
\begin{align}
  \begin{bmatrix}
    \Omega_{1}\\
    \vdots\\
    \Omega_{\tilde n}
  \end{bmatrix} =  \left(V\otimes I_m \right)^+\begin{bmatrix}
      H_{0}\\
      \cdots\\
      H_{3\tilde n-2}
    \end{bmatrix}.
\end{align}
Hence, by Lemma~\ref{lemma:functionclass}.3, $\Omega_{k,i} - \Omega_i\sim\mathcal C(-\gamma)$.

\subsection*{Convergence of $\hat \varphi_k$, $\hat \vartheta_k$ and $\mathcal W_k$}
First we need to prove that $\hat \varphi_k-\varphi_k\sim \mathcal C(-\gamma)$, which holds as long as $\hat \varphi_{k,i}-\varphi_{k,i}\sim \mathcal C(-\gamma)$ for all $i$, where
\begin{align*}
   \varphi_{k,i} = \lambda_{i } \varphi_{k-1,i} + \Omega_{i} \phi_{k},\,\varphi_{-1,i}=0.
\end{align*}

Notice that the error between $\hat \varphi_{k,i}$ and $\varphi_{k,i}$ satisfies the following recursive equation:
  \begin{align*}
    \varphi_{k+1,i}-\hat\varphi_{k+1,i} &= (\lambda_i-\lambda_{k,i})\varphi_{k,i}+\lambda_{k,i}(\varphi_{k,i}-\hat\varphi_{k,i})  \\
    &+ (\Omega_i-\Omega_{k,i})\phi_k.
  \end{align*}
  For any $\epsilon > 0$, we have
  \begin{align*}
    \frac{\|\varphi_{k+1,i}-\hat\varphi_{k+1,i}\|}{(k+1)^{-\gamma+2\epsilon}} &\leq |\lambda_{k,i}|\frac{\|\varphi_{k,i}-\hat\varphi_{k,i}\|}{k^{-\gamma+2\epsilon}}  +\frac{|\lambda_i-\lambda_{k,i}|}{k^{-\gamma+\epsilon}}\frac{\|\varphi_{k,i}\|}{k^{\epsilon}}\\
    &+ \frac{\|\Omega_i-\Omega_{k,i}\|}{k^{-\gamma+\epsilon}}\frac{\|\phi_k\|}{k^\epsilon}.
  \end{align*}

  Notice that $\phi_k = U_k^{1/2}\zeta_k$. Since $\zeta_k\sim \mathcal C(0)$ by Lemma~\ref{lemma:functionclass}.6, and $U_k$ is upper bounded by Lemma~\ref{lemma:Ubound}, $\phi_k\sim \mathcal C(0)$. Thus, $\varphi_{k,i}\sim \mathcal C(0)$ by Lemma~\ref{lemma:functionclass}.4. Furthermore, since $\lambda_{k,i}-\lambda_i\sim \mathcal C(-\gamma)$ and $\Omega_{k,i}-\Omega_i\sim \mathcal C(-\gamma)$ , for any $\epsilon_1 > 0$, there exists $K$ (possibly random), such that for any $k \geq K$, the following inequalities hold almost surely,
  \begin{align*}
  |\lambda_i-\lambda_{k,i}|\leq \epsilon_1,\, \frac{|\lambda_i-\lambda_{k,i}|}{k^{-\gamma+\epsilon}}\frac{\|\varphi_{k,i}\|}{k^{\epsilon}}+ \frac{\|\Omega_i-\Omega_{k,i}\|}{k^{-\gamma+\epsilon}}\frac{\|\phi_k\|}{k^\epsilon}\leq \epsilon_1.
  \end{align*}
   Therefore, for $k\geq K$, we have
  \begin{align*}
    \frac{\|\varphi_{k+1,i}-\hat\varphi_{k+1,i}\|}{(k+1)^{-\gamma+2\epsilon}}  &\leq (|\rho|+\epsilon_1)\times \frac{\|\varphi_{k,i}-\hat\varphi_{k,i}\|}{k^{-\gamma+2\epsilon}} +  \epsilon_1.\;a.s.
  \end{align*}
  Now since $|\rho|<1$, we can choose $\epsilon_1$ small enough such that $|\rho|+\epsilon_1 < 1$, therefore,
  \begin{align*}
    \limsup_{k\rightarrow\infty}\frac{\|\varphi_{k,i}-\hat\varphi_{k,i}\|}{k^{-\gamma+2\epsilon}}\leq \frac{\epsilon_1}{1-|\rho|-\epsilon_1}.\;a.s.
  \end{align*}
Hence, $\|\varphi_{k,i}-\hat\varphi_{k,i}\|/k^{-\gamma+3\epsilon}\asrightarrow 0$, which proves that $\varphi_{k,i}-\hat\varphi_{k,i}\sim\mathcal C(-\gamma)$.
\begin{lemma}
\begin{align}
 \frac{1}{k+1}\sum_{t=0}^k \vartheta_t \vartheta_t^T - \mathcal W \sim \mathcal C(-0.5),
\end{align}
where $\vartheta_k \triangleq \sum_{t=0}^{k}  CA^{t} w_{k-t}+v_{k} + CA^{k+1} x_{-1}$.
\end{lemma}
\begin{proof}
  Let us define function $\mathcal A:\mathbb R^{n\times n}\rightarrow\mathbb R^{n\times n}$, such that for any symmetric matrix $X\in \mathbb S^{n\times n}$,
  \begin{align*}
    \mathcal A(X) = X + AXA^T + A^2XA^{2T}+\cdots,
  \end{align*}
  For non-symmetric $X\in\mathbb R^{n\times n}$, we define
  \begin{align*}
    \mathcal A(X) = \mathcal A\left(\frac{X+X^T}{2}\right).
  \end{align*}
  One can prove that
  \begin{align*}
    \mathcal A(X) - A^k\mathcal A(X)A^{kT} = \sum_{i=0}^{k-1} A^i\frac{X+X^T}{2}A^{iT}.
  \end{align*}
  To simplify notations, let us define $w_{-1} = x_{-1}$. By mathematical induction, $\sum_{t=0}^k\vartheta_t\vartheta_t^T$ can be written as
  \begin{align}
    \sum_{t=0}^k\vartheta_t\vartheta_t^T=\mathcal M_k - CA\mathcal N_k A^TC^T,
  \end{align}
  where
  \begin{align}
    \mathcal M_k &= \mathcal M_{k-1} + \Pi_k\\
    \mathcal N_k &= A\mathcal N_{k-1}A^T+2 \mathcal A\left(\left(\sum_{t=-1}^kA^{k-t}w_t\right)w_k^T\right) - \mathcal A(w_kw_k^T).
  \end{align}
  with
\begin{align*}
 \Pi_k &=   v_kv_k^T +  v_k\left(\sum_{t=-1}^k CA^{k-t}w_t\right)^T+\left(\sum_{t=-1}^k CA^{k-t}w_t\right)v_k^T\nonumber\\
    &+ 2C\mathcal A\left(\left(\sum_{t=-1}^{k} A^{k-t}w_t\right)w_k^T\right)C^T - C\mathcal A\left(w_kw_k^T\right)C^T,
\end{align*}
and initial condition 
\begin{align}
  \mathcal N_{-1} = \mathcal A\left(x_{-1}x_{-1}^T\right),\, \mathcal M_{-1} = CA\mathcal N_{-1}A^TC^T.
\end{align}

One can then prove that
\begin{align*}
  \mathbb E(\Pi_k|\mathcal F_{k-1}) = \mathcal W,\, \mathbb E\|\Pi_k-\mathcal W\|^2\sim O(1).
\end{align*}
Hence, $\mathcal M_k - k\mathcal W$ is a martingale and $\mathcal M_k/k -\mathcal W \sim \mathcal C(-0.5)$ by Lemma~\ref{lemma:MatrixStrongLaw}. On the other hand, for $\mathcal N_k$, since $A\otimes A$ is stable, $\mathcal N_k\sim \mathcal C(0)$ by Lemma~\ref{lemma:functionclass}, which proves that 
\begin{align*}
  \frac{1}{k+1}\sum_{t=0}^k \vartheta_t\vartheta_t^T - \mathcal W \sim \mathcal C(-0.5).
\end{align*}
\end{proof}

Now we can rewrite $\mathcal W_k - \mathcal W$ as
\begin{align*}
  \mathcal W_k &- \mathcal W = \left(\frac{1}{k+1}\sum_{t=0}^k \vartheta_t\vartheta_t^T   - \mathcal W\right)  \\
               &- \frac{1}{k+1}\sum_{t=0}^k\left( \vartheta_t (\hat \phi_t - \phi_t)^T + (\hat \phi_t - \phi_t)\vartheta_t^T\right)\\
  &+ \frac{1}{k+1}\sum_{t=0}^k(\hat \phi_t - \phi_t)(\hat \phi_t - \phi_t)^T,
\end{align*}
By Lemma~\ref{lemma:functionclass}, $\mathcal W_k-\mathcal W \sim \mathcal C(\max\{-0.5,-\gamma,-2\gamma\}) = \mathcal C(-\gamma)$.
\subsection*{Convergence of the Rest}
By Lemma~\ref{lemma:functionclass}.3, one can prove that $\mathcal P_k-\mathcal P$, $\mathcal X_k-\mathcal X$ are all of the class $\mathcal C(-\gamma)$, as they are differentiable functions of $\lambda_{k,i}$, $\Omega_{k,i}$ and $\mathcal W_k$. Therefore, $ U_{k,*}-U_*\sim \mathcal C(-\gamma)$ since $U_{k,*}$ is a differentiable function of $\mathcal P_k$ and $\mathcal X_k$ at a neighborhood of $\mathcal P$ and $\mathcal X$ (see \cite{polynomialroots}).

Hence, one can prove that $\mathcal U_k - \mathcal U\sim \mathcal C(-\gamma)$, as $\mathcal U_k$ is a differentiable function of $\lambda_{k,i}$, $\Omega_{k,i}$ and $U_{k,*}$.

Finally we prove that $\hat g_k-g_k\sim\mathcal C(-\gamma)$ due to Lemma~\ref{lemma:functionclass}.1.
\bibliographystyle{IEEEtran}
\bibliography{reference}
\end{document}